\documentclass[oneside,english]{amsart}
\pdfoutput=1
\usepackage[T1]{fontenc}
\usepackage[latin9]{inputenc}
\usepackage{geometry}
\usepackage{verbatim}
\usepackage{amsthm}
\usepackage{amstext}
\usepackage{amssymb}
\usepackage{graphicx}
\usepackage{esint}
\usepackage{url}
\usepackage{yhmath}
\usepackage[all,cmtip]{xy}
\usepackage{color}

\usepackage[percent]{overpic}

\usepackage[space]{grffile} % This allows file names with spaces
\makeatletter
\numberwithin{equation}{section}
\numberwithin{figure}{section}
\theoremstyle{plain}
\newtheorem*{thm*}{\protect\theoremname}
\theoremstyle{plain}
\newtheorem{thm}{\protect\theoremname}[section]
\theoremstyle{plain}
\newtheorem{lem}[thm]{\protect\lemmaname}
\theoremstyle{remark}

\theoremstyle{plain}
\newtheorem{prop}[thm]{\protect\propositionname}
\theoremstyle{plain}

\theoremstyle{plain}

\theoremstyle{plain}

\theoremstyle{plain}

\theoremstyle{plain}

\theoremstyle{plain}

\newtheorem{keylem}[thm]{Key lemma}
\theoremstyle{plain}

\usepackage{mathrsfs}
\usepackage{subfigure}

\usepackage{etex}
\newcommand{\SLE}{\mathrm{SLE}}

\newcommand{\PR}{\mathbb{P}}
\newcommand{\EX}{\mathbb{E}}

\newcommand{\bR}{\mathbb{R}}
\newcommand{\R}{\bR}

\newcommand{\bZ}{\mathbb{Z}}
\newcommand{\Z}{\bZ}
\newcommand{\bN}{\mathbb{N}}
\newcommand{\N}{\bN}

\newcommand{\bC}{\mathbb{C}}
\newcommand{\C}{\bC}

\newcommand{\bH}{\mathbb{H}}

\newcommand{\domain}{\Lambda}
\newcommand{\bdry}{\partial}

\newcommand{\confmap}{\phi} % Conformal map to the unit disk
\newcommand{\confmapD}{\confmap^{-1}} % Conformal map from the unit disk
\newcommand{\confmapDH}{\boldsymbol{\psi}} % A fixed conf map from D to H
\newcommand{\confmapHD}{\boldsymbol{\psi}^{-1}} % A fixed conf map from D to H
\newcommand{\confmapBdry}{\varphi} % A conf map from (domain, a, b) to (D, -1, 1), i.e., normalized by 2 bdry pts.

\newcommand{\ctsfcns}{\mathcal{C}} % space of cts fcns

\newcommand{\set}[1]{\left\{  #1\right\}  }

\newcommand{\clos}[1]{\overline{ #1 }}

\newcommand{\HarmMeas}{\mathsf{H}}

\newcommand{\diam}{\mathrm{diam}}

\newcommand{\eps}{\epsilon}
\newcommand{\Area}{\mathrm{Area}}
\newcommand{\UnitD}{\mathbb{D}}

\newcommand{\Cara}{\stackrel{\mathrm{Cara}}{\longrightarrow}}
\newcommand{\domainn}{\domain_*}
\newcommand{\n}{*}
\renewcommand{\curve}{\eta}
\newcommand{\calA}{\mathcal{A}}
\newcommand{\calB}{\mathcal{B}}

\newcommand{\bbH}{\mathbb{H}}
\newcommand{\bbP}{\mathbb{P}}

\makeatletter
\newcommand*{\centerfloat}{%
  \parindent \z@
  \leftskip \z@ \@plus 1fil \@minus \textwidth
  \rightskip\leftskip
  \parfillskip \z@skip}
\makeatother

\AtBeginDocument{

}

\makeatother

\usepackage{babel}
\providecommand{\corollaryname}{Corollary}
\providecommand{\lemmaname}{Lemma}
\providecommand{\propositionname}{Proposition}
\providecommand{\remarkname}{Remark}
\providecommand{\theoremname}{Theorem}
\providecommand{\conjecturename}{Conjecture}

\begin{document}

\

\vspace{2.5cm}

\begin{center}
\LARGE \bf \scshape {Limits of conformal images and conformal images of limits for planar random curves}
\end{center}

\vspace{0.75cm}

\begin{center}
{\large \scshape Alex M. Karrila}\\
{\footnotesize{\tt alex.karrila@abo.fi; alex.karrila@gmail.com}}\\
{\small{\AA bo Akademi Matematik}}\\
{\small{Henriksgatan 2, 20500 \AA bo, Finland}}%\bigskip{}
% \bigskip{}

\end{center}

\vspace{0.75cm}

\begin{center}
\begin{minipage}{0.85\textwidth} \footnotesize
{\scshape Abstract.}
We address the scaling limits of random curves arising from, e.g., planar lattice models, especially in rough domains. The well-known precompactness conditions of Kemppainen and Smirnov show that certain crossing probability estimates guarantee the subsequential weak convergence of the random curves in the topology of unparametrized curves, as well as in a topology inherited from the unit disc via conformal maps. We complement this result by proving that proceeding to weak limit commutes with changing topology, i.e., limits of conformal images are conformal images of limits, with minimal boundary regularity assumptions on the domains where the random curves lie. Such rough boundaries are especially interesting if, in the context of multiple random curves, a limit candidate is defined in terms of iterated SLE type processes with $\kappa \in (4, 8)$, and one hence needs to study (boundary-touching) curves in domains slit by other random curves.
\\ \\
\textbf{MSC classes:} Primary: 60Dxx; Secondary: 60J67, 82B20.
\end{minipage}
\end{center}

\vspace{0.75cm}
\tableofcontents

%\bigskip{}

\section{Introduction}

%CFT
%
%SLE; defn via conf inv (diagram); convergence results
%
%Precompactness + identification. Simplification K-S
%
%"Players"; commutative diagrams.
%
%"Missing link". Trivial when uniformly locally connected approximations. Improving: independent interest + multiple SLEs
%
%Organization of paper.
%
%Acknowledgements

\subsection{Background}

In physics, Conformal field theories appear as scaling limit candidates for planar lattice models of statistical mechanics at criticality~\cite{Polyakov, BPZ-infinite_conformal_symmetry_in_2D_QFT, BPZ-infinite_conformal_symmetry_of_critical_fluctiations, Cardy-conformal_invariance_and_statistical_mechanics}.
Mathematical-physics proofs of conformal invariance properties in such scaling limits have only been achieved more recently, 
one successful approach being to prove the convergence of interface curves to conformally invariant random curves, called Schramm--Loewner evolutions (SLEs)~\cite{Schramm-LERW_and_UST, RS-basic_properties_of_SLE}. Such convergence proofs have been established in several lattice models~\cite{Smirnov-critical_percolation, LSW-LERW_and_UST, SS05, CN07, Zhan-scaling_limits_of_planar_LERW, SS09,  CDHKS-convergence_of_Ising_interfaces_to_SLE}. Also different variants of SLEs~\cite{LSW-Conformal_restriction, Dubedat-duality, Zhan-scaling_limits_of_planar_LERW} and multiple SLEs~\cite{BBK-multiple_SLEs, Dubedat-commutation, LK-configurational_measure, PW, BPW} have been introduced, and shown to serve as scaling limits~\cite{HK-Ising_interfaces_and_free_boundary_conditions, Izyurov-critical_Ising_interfaces_in_multiply_connected_domains, KS-FK_Hypegeometric_SLE, BPW, GW18, mie3}.

The most typical route to an SLE convergence proof consists of two parts: precompactness, i.e., the existence of subsequential scaling limits, and identification of any subsequential limit. The precompactness part is usually not model-specific, in the sense that it is deduced by verifying certain crossing probability estimates~\cite{AB-regularity_and_dim_bounds_for_random_curves, KS} for the model of interest. These estimates for instance guarantee that the random curve laws on lattices of increasingly fine mesh are precompact in a standard topology of unparametrized curves, as well as in a topology inherited from curves on the unit disc $\UnitD$ via conformal maps. The main theorem~\ref{thm: extension of Kemppainen's precompactness} of this paper states that the weak limits in these two topologies agree, i.e., limits of conformal images are conformal images of limits, assuming the same crossing estimates as~\cite{KS} and imposing essentially no boundary regularity assumptions.

\begin{figure}
\includegraphics[width=0.35\textwidth]{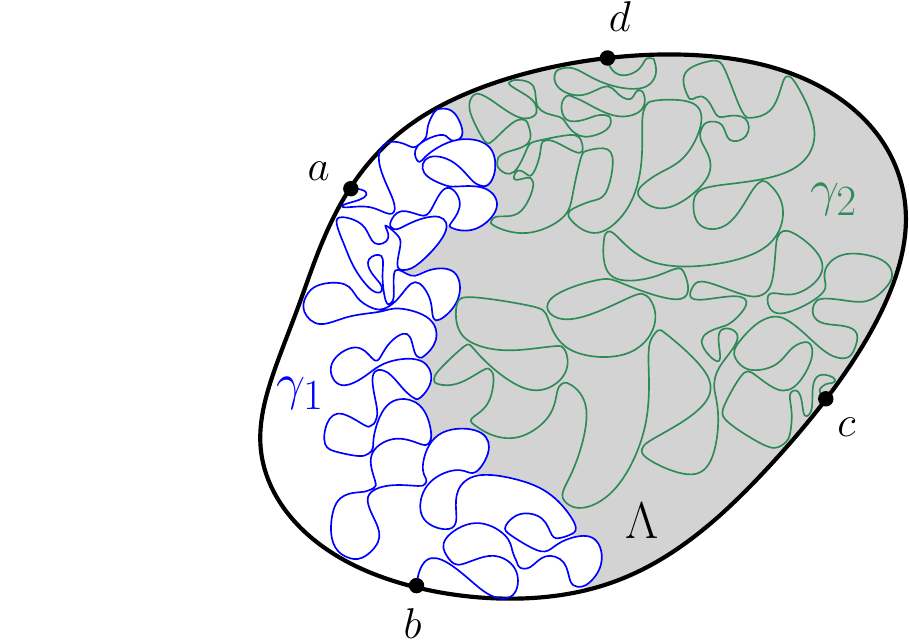}
\qquad
\includegraphics[width=0.35\textwidth]{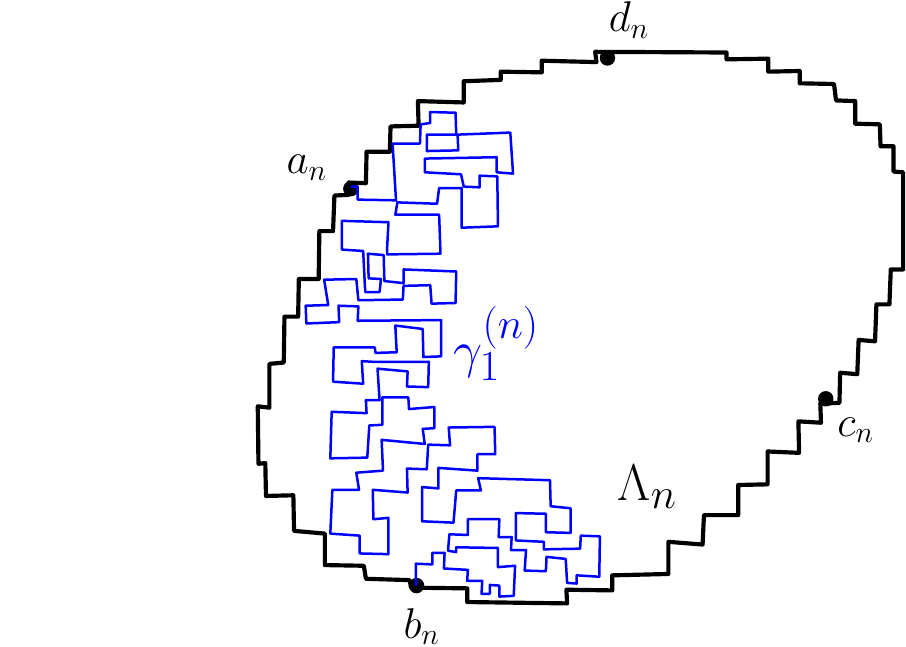}
\caption{
\label{fig:schematic application}
A schematic illustration of a situation where irregular boundary approximations naturally arise: For example multiple SLE type curves (left panel) can be defined in terms of iterated SLE processes, so that the second curve is sampled in the area left after sampling the first curve (shaded). The curves are non-self-crossing but for $\kappa \in (4, 8)$ they intersect themselves and each other in a fractal manner. In a discrete iterated sampling (right panel), the corresponding component of $\domain_n \setminus \gamma^{(n)}_1$ may now have ``deep fjords'', even if the approximating domains $\domain_n$ are well-behaved. Based on the present note, no \textit{a priori} results on $\gamma^{(n)}_1$ need to be proven to exclude such fjords.
}
\end{figure}

The corresponding result is easily deduced if either suitable boundary regularity is assumed, or if boundary visits of the random curves can be excluded (see Section~\ref{subsec:easy} for details), which is probably why it has not, to the best of our knowledge, been priorly explicated. Apart from independent interest, studying domains with irregular boundaries is motivated by, e.g., multiple SLEs type curves with $\kappa \in (4, 8)$, 
due to their definition in terms of iterated SLE type growth processes, see Figure~\ref{fig:schematic application}. The results in this paper do not rely on any underlying lattice structure, and they also have interesting implications in terms of chordal SLEs, which will be discussed.

\subsection{The main application}
\label{subsec: main application}

All the basic SLE type curves are defined via a driving function $W: \R_{\ge 0} \to \R$, $t \mapsto W_t$, which first yields via a Loewner type equation a curve in a reference domain, say the curve $\gamma_\UnitD$ in the unit disc $\UnitD$, and the SLE type curve $\gamma$ in the domain of interest $\domain$ is then defined via conformal invariance:
\begin{align*}
W \stackrel{ \text{Loewner} }{\longrightarrow} \gamma_\UnitD \stackrel{ \text{conformal} }{\longrightarrow} \gamma.
\end{align*}
Now, assume that we study a lattice model on graph domains $\domain_n$, $n \in \N$, approximating the domain $\domain$ in some sense, and that we wish to show that some discrete interface curves $\gamma^{(n)}$ on $\domain_n$ converge weakly to the SLE type curve $\gamma$ in $\domain$. Proceeding with the strategy above, the results of~\cite{AB-regularity_and_dim_bounds_for_random_curves, KS} often guarantee the subsequential convergence of $\gamma^{(n)}$ to some weak limit $\gamma$. Due to the very definition of SLE, the most typical and straightforward way to identify this limit as an SLE (cf.~\cite{KS}) is to first map it to a driving function,
\begin{align*}
\gamma^{(n)} \stackrel{ \text{conformal} }{\longrightarrow} \gamma^{(n)}_\UnitD \stackrel{ \text{Loewner} }{\longrightarrow} W^{(n)},
\end{align*}
and then prove precompactness also in the sense of the curves $\gamma^{(n)}_\UnitD$ and the functions $W^{(n)}$, and finally show, informally speaking, that the diagram 
\[
\xymatrix@C=3.5em{
\gamma^{(n)} \ar[d]^{n \to \infty} \ar[r]^{\text{conformal}} & \gamma_{\UnitD }^{(n)} \ar[d]^{ n \to \infty } \ar[r]^{\text{Loewner}}  & W^{(n)}  \ar[d]^{n \to \infty}  \\
\gamma & \gamma_{\mathbb{D}} \ar[l]^{\text{conformal}}  & W \ar[l]^{\text{ Loewner }}
}
\]
with mappings and weak convergences, commutes. The main result of this paper addresses the bottom left horizontal arrow in the diagram above (in the case of minimal boundary regularity assumptions and boundary-visiting curves), while the other ones follow from the results of~\cite{AB-regularity_and_dim_bounds_for_random_curves, KS}.

%\subsection{Organization} In Section~\ref{sec: preli}, we define the necessary notions of convergence and recall some basic facts about SLE. In Section~\ref{sec: trivial and warning examples} we discuss easy special cases of our main theorem, when either the boundaries of the domains $\domain_n$ or the limiting curves $\gamma_\UnitD$ are regular enough, and give warning examples showing that the general case will require caution. These short computations should explain why the main question of this paper has not attracted attention in the context of single SLE convergence, and why it should do so in the context of multiple SLEs. Section~\ref{sec: main thm} constitutes the statement and parts of the proof of the main theorem~\ref{thm: extension of Kemppainen's precompactness}, while the technical key step of the proof is postponed to Appendix~\ref{sec: pf of key lemma}. Some extensions of Theorem~\ref{thm: extension of Kemppainen's precompactness} are given in Section~\ref{subsec: further extensions}. Finally, in Section~\ref{sec: applications} we demonstrate the power and nontriviality of our main theorem by giving proofs of certain regularity and stability properties of the chordal SLE.

\textbf{Acknowledgements.} The author is very grateful to the anonymous referee of this note, whose remarks greatly helped to improve several proofs, as well as the presentation. The author also wishes to thank Dmitry Chelkak, Christian Hagendorf, Eveliina Peltola, Fredrik Viklund, and especially Antti Kemppainen and Kalle Kyt\"{o}l\"{a} for useful comments and discussions. The author was financially supported by the Vilho, Yrj\"{o} and Kalle V\"{a}is\"{a}l\"{a} Foundation and by the Academy of Finland (grant \#339515).

\section{Preliminaries}

\subsection{Conformal structures}
\label{subsubsec: Cara convergence}

The \emph{Riemann (uniformization) map} from a simply-connected domain $\domain \subsetneq \C$ to $\UnitD$ normalized at $u \in \domain$ is the unique conformal map $\confmap: \domain \to \UnitD$ such that $\confmap(u)=0$ and $\confmap'(u) \in \R_+$. 
A sequence of simply-connected domains $( \domain_n )_{n \in \N}$ converges to $\domain$ with respect to $u$ in the \textit{Carath\'{e}odory sense}, denoted $\domain_n \Cara \domain$, if $u$ lies in all the domains and the inverse Riemann maps $\confmap_n^{-1}$ as above converge uniformly on compact subsets of $\UnitD$ to $\confmap^{-1}$, where all the conformal maps are normalized at the same point $u$. By Carath\'{e}odory's kernel theorem~\cite[Theorem~1.8]{Pommerenke-boundary_behaviour_of_conformal_maps}, this occurs if and only if every $z \in \domain$ has some neighbourhood $V_z$ such that $V_z \subset \domain_n$ for all large enough $n$, and for every $p \in \bdry \domain$ there exists a sequence $p_n \in \bdry \domain_n$ such that $p_n \to p$. Note that whether or not a convergence $\domain_n \Cara \domain$ occurs is hence independent of the choice of $u \in \domain$ (only taking the tail of the sequence $( \domain_n )_{n \in \N}$ if needed to ensure $u \in \domain_n$).

The boundary behaviour of conformal maps is a central issue in this note. A \textit{cross cut} of a simply-connected domain $\domain$ is an open Jordan arc $S$ in $\domain$ such that $\overline{S} = S \cup \{ a, b \}$, where $a, b \in \bdry \domain$; a \textit{null chain} is a sequence $(S_n)_{n \in \N}$ of disjoint cross cuts, nested in the sense that for all $n$, $S_n$ separates $S_{n+1}$ from $S_0$ in $\domain$, and such the the spherical-metric diameter of $S_n$ tends to zero as $n \to \infty$. The intuitive notion of a ``boundary point'' often needs to be replaced by a \textit{prime end}, i.e., an equivalence class of null chains of $\domain$ when two null chains $(S_n)_{n \in \N}$ and $(\tilde{S})_{n \in \N}$ are equivalent if for any $m$ large enough there exists $n$ such that the cross cut $S_m$ separates $\tilde{S}_n$ from $\tilde{S}_0$ in $\domain$, and $\tilde{S}_m$ separates ${S}_n$ from ${S}_0$ in $\domain$. The prime end theorem by Carath\'{e}odory~\cite[Theorem~2.15]{Pommerenke-boundary_behaviour_of_conformal_maps} then states that a conformal map $\confmap: \domain \to \UnitD$ induces a bijection between the prime ends of $\domain$ and the unit sphere $ \bdry \UnitD$, such that if  $(S_n)_{n \in \N}$ is a representative of a prime end $p$, then $\confmap(S_n)$ converge to the image point of $p$ on $\bdry \UnitD$. We extend the definition of Carath\'{e}odory convergence to domains with finitely many marked prime ends using this bijection: e.g., $(\domain_n; p^{(n)}) \Cara (\domain; p )$ if $\domain_n \Cara \domain$ and $\confmap_n (p^{(n)})  \to   \confmap (p)$ on $\bdry \UnitD$, where we also denoted the induced prime ends bijections by $\confmap_n$ and $\confmap$, respectively.

A less general but arguably more direct boundary extension of a conformal map $\confmap^{-1}: \UnitD \to \domain$ to $\bdry \UnitD$ is obtained via radial limits. Formally, for $0 < \epsilon \leq 1$, denote by $P_\eps : \overline{\UnitD} \to \UnitD$ the radial projection
\begin{align*}
P_\eps (z) := \tfrac{z}{ \vert z \vert } \min \{ 1- \eps, \vert z \vert \}.  
\end{align*}
Then, the \textit{radial limit} of $\confmap^{-1}$ at $z \in \bdry \UnitD$ is
given by $\lim_{\eps \downarrow 0} \confmap^{-1} \circ P_\eps (z)$;
by the classical Fatou theorem this limit indeed exists for Lebesgue almost every $z \in \bdry \UnitD$ if $\domain$ is bounded. While the definition above involves the conformal map $\confmap^{-1}$, both the existence and the value of a radial limit at the point on $\bdry \UnitD$ that corresponds to a given prime end $p$ are actually properties of only the domain $\domain$ and 
 $p$ that do not depend on the choice of the conformal map;\footnote{This essentially stems from the fact that the conformal (M\"{o}bius) maps $\UnitD \to \UnitD$ are conformal and differentiable also on the boundary $\partial \UnitD$, and a radial (i.e., boundary-normal) line segment maps to a radial line segment up to a second-order correction; see, e.g.,~\cite[Corollary~2.17]{Pommerenke-boundary_behaviour_of_conformal_maps} for a formal proof.} we thus simply say that $p$ \textit{has radial limits}.

A \textit{topological quadrilateral}, or a \textit{quad}, $(Q; S_0, S_1, S_2, S_3)$ consists of a planar domain $Q$ homeomorphic to a square and arcs $S_0, S_1, S_2, S_3$ of its boundary, indexed counter-clockwise and corresponding to the edges of the square. We will denote $(Q; S_0, S_1, S_2, S_3)$ by simply $Q$ if the sides are clear in the context. The conformal structure of a quad is captured by the \textit{modulus} $m(Q; S_0, S_1, S_2, S_3)$ (also called the extremal distance of $S_0$ and $S_2$ in $Q$); it is the unique $L > 0$ such that there exists a conformal map from $Q$ to the rectangle $(0, L) \times (0,1)$, so that the sides $S_0, S_1, S_2, S_3$ of $Q$ correspond to the edges of the rectangle, and $S_0$ to $\{ 0 \} \times [0,1]$ (see e.g.~\cite[Chapter~4]{Ahlfors}).

\subsection{Probability measures on metric spaces}
\label{subsec:proba on metric spaces}

Let $(Y, d_Y)$ denote a metric space and $\calB_Y$ the corresponding Borel sigma algebra. A sequence $\mu_n$ of probability measures on $(Y, \calB_Y)$ \textit{converges weakly} to the measure $\mu$ on $(Y, \calB_Y)$  if $\mu_n[f] \to \mu[f]$ for all bounded continuous functions $f: Y \to \R$.
A collection $\{ \mu_\alpha \}_{\alpha \in \calA}$ of probability measures on $(Y, \calB_Y)$ is \textit{precompact (in the topology of weak convergence)} if every sequence $\mu_n$ from that collection contains a subsequence that converges weakly. The collection is \textit{tight} if for every $\epsilon > 0$ there exists a compact set $K \subset Y$ such that $\mu_\alpha [K] > 1 - \eps$ for all $\alpha \in \calA$. By Prokhorov's theorem a tight collection is precompact, and if $(Y, d)$ is separable and complete, then tightness and precompactness are equivalent.

The most prominent non-trivial metric space in this note is the \textit{space of unparametrized plane curves} $X(\C)$.
A planar \emph{curve} is a continuous mapping $\gamma: [0,1] \to \C$.
We define the equivalence relation $\sim$ on curves by setting $\gamma \sim \tilde{\gamma}$, if
\begin{align*}
\inf_{\psi} \bigg(  \sup_{t \in [0,1]} \vert \gamma(t) - \tilde{\gamma} \circ \psi (t) \vert  \bigg) = 0,
\end{align*}
where the infimum is taken over increasing (hence continuous) bijections $ \psi: [0,1] \to [0,1]$, and $X(\C)$ is the space of these equivalence classes. We will always study an equivalence class through a representative, and we will not make notational difference between a curve, its equivalence class, or its trace. The space $X(\C)$ is equipped with the metric
\begin{align*}
d(\gamma, \tilde{\gamma}) = \inf_{\psi} \bigg(  \sup_{t \in [0,1]} \vert \gamma(t) - \tilde{\gamma} \circ \psi (t) \vert  \bigg),
\end{align*}
where the infimum is again over increasing bijections $ \psi: [0,1] \to [0,1]$ (and hence independent of the choice of representatives).

The (closed) subset of $X(\C)$ consisting of curves that stay in $\clos{\UnitD}$ is denoted by  $X(\clos{\UnitD})$. By the closedness, in particular, if $\mu_n$ are supported on $X(\clos{\UnitD})$ and $\mu_n \to \mu$ weakly on $X(\clos{\UnitD})$ (resp. weakly on $X(\C)$), then  also $\mu_n \to \mu$ on $X(\C)$, \textit{a fortiori} (resp. $\mu$ is supported on and the weak convergence also holds on $X(\clos{\UnitD})$, by the Portmanteau theorem).

The space $\ctsfcns$ of continuous functions $W_\cdot : \R_{\ge 0} \to \R$ is studied in the metric of uniform convergence over compact subsets:
\begin{align}
\label{eq: metric of uniform convergence over compact subsets for functions of reals}
%d(W^{(1)}, W^{(2)}) = \sum{n \in \N} 2^{-n} \max \{ 1, \sup_{t \in [0,n]} \vert W^{(2)}_t - W^{(1)}_t \vert \}.
d(W, \tilde{W}) = \sum_{n \in \N} 2^{-n} \min \{ 1, \sup_{t \in [0,n]} \vert \tilde{W}_t - W_t \vert \}.
\end{align}
The spaces $X(\C)$ and $\ctsfcns$ are both complete and separable (which is essentially inherited from the space of continuous functions on $[0, 1]$ with the sup norm); in particular, Prokhorov's theorem applies.

\subsection{(Schramm--)Loewner evolutions}

We briefly recap Loewner evolutions and SLEs; note that the contribution and proofs of this paper actually have little to do with them, even if they are vital to understand the context.

\subsubsection{\textbf{The deterministic case}}

The \textit{Loewner differential equation} in the upper half-plane $\bH$ determines a family of complex analytic mappings $g_t$, $t \ge 0$ by $g_0 (z) = z \in \bH$ and
\begin{align}
\label{eq: Loewner ODE}
\partial_t g_t (z) &= \frac{2}{g_t(z) - W_t}, 
\end{align}
where $W_\cdot:\R_{ \ge 0} \to \R$ is a given continuous function, called the \textit{driving function}.
For a given $z \in \bH$, the solution $g_t(z)$ is defined up to a possibly infinite explosion time $\tau(z)$ when $W_\cdot$ and $g_\cdot (z)$ first collide.
The set where the solution is not defined is denoted by
$
K_t =  \{ z \in \bH : \tau(z) \le t  \} .
$
The sets $K_t$ are growing in $t$, and for all $t$, the are \textit{hulls}, i.e., $K_t$ are bounded, closed in $\bH$, and $H_t =: \mathbb{H} \setminus K_t $ is simply-connected. Furthermore, $g_t$ is a conformal map $H_t \to \bH$ such that $g_t (z) = z + \frac{2t}{z} + O(z^{-2}) $ as $z \to \infty.$

Conversely, for a family of growing hulls $(K_t)_{t \geq 0}$, there exists a driving function $W_\cdot$ so that $K_t$, up to time re-parametrization, are obtained as the Loewner hulls of this function, if $K_t$ satisfy the so-called \textit{local growth condition} and and their \textit{half-plane capacity} tends to infinity as $t \uparrow \infty$; see, e.g.,~\cite{Lawler-SLE_book} for details. Such families $(K_t)_{t \geq 0}$ are called \textit{Loewner growth processes}, and the Loewner differential equation thus produces a bijection between Loewner growth processes and continuous functions $W: \R_{\geq 0} \to \R$.

We say that a Loewner growth process $(K_t)_{t \ge 0}$ is \textit{generated} by a continuous map $\gamma_{\bH}:\R_{ \ge 0} \to \overline{\bH}$ if $H_t$ is the unbounded component of $\mathbb{H} \setminus \gamma_{\bH}([0,t])$ for all $t$. Whether a given $\gamma_{\bH}$ (up to re-parametrization) generates some Loewner growth process can be determined based on the aforementioned conditions; in particular, the local growth condition always holds if $\gamma_{\bH}$ is a simple chordal curve staying in $\bH$ except for the end points.

Throughout this paper, we fix conformal (M\"{o}bius) maps $\confmapDH$ and $\confmapHD$ taking the domain $(\UnitD; -1, 1)$ with marked prime ends to $(\bH; 0, \infty)$ and vice versa, say for definiteness
\begin{align*}
\confmapDH (z) = i \tfrac{z+1}{1-z}, \qquad \confmapHD (z) = \tfrac{z - i}{z + i}.
\end{align*}
If the hulls corresponding to a driving function $W$ are generated by $\gamma_{\bH}$ and $\gamma_{\bH} (t) \stackrel{t \to \infty}{\longrightarrow} \infty$, we define the curve $\gamma_\UnitD: [0, 1]  \to \overline{ \UnitD }$ by $\gamma_\UnitD (t) = \confmapHD (\gamma_{\bH} ( \frac{t}{1-t} ))$ for $t \in [0,1)$ and $\gamma_\UnitD(1) = 1$ and say that $W_t$ and $\gamma_\UnitD $ are \emph{deterministic Loewner transforms} of each other.

\subsubsection{\textbf{Random Loewner processes}}
\label{subsubsec:rand growth proc}

We equip the space of Loewner growth processes with the metric (topology) of their driving functions $W \in \ctsfcns$. A random Loewner growth process on some probability space is then a growth process valued random variable (i.e., measurable in our topology on $\ctsfcns$). E.g., if $B_t$ is a standard Brownian motion on the probability space $(\Omega, \mathcal{F}, \bbP)$, then the \textit{chordal Schramm--Loewner evolution from $0$ to $\infty$ in $\overline{\bH}$ with parameter $\kappa \ge 0$} ($\SLE(\kappa)$ for short) is the random Loewner growth process with $W_t = \sqrt{\kappa} B_t$. We say that a random curve $\gamma_\UnitD \in X(\overline{\UnitD})$ and a driving function $W \in \ctsfcns$ defined on the same probability space are \textit{(probabilistic) Loewner transforms} of each other if they are almost surely deterministic Loewner transforms of each other. Discrete models with finitely many possible curve trajectories often trivially give rise to such a random Loewner transform pair; the existence of a suitable random variable $\gamma_\UnitD$ for the SLE is less trivial but is recalled in Section~\ref{subsec:SLE as rand curve}.

\section{Some easy special cases and warning examples}

\subsection{Easy special cases}
\label{subsec:easy}

We point out two common special cases where the main application of our results, as explained in the Introduction, is readily handled. The core of both cases is that the uniform convergence ``$\confmap_n^{-1} \to \confmap^{-1}$'' extends to (parts of) the boundary; not even conformality is actually needed. We however stick to the context of conformal maps for simplicity.

\subsubsection{\textbf{Regular boundary approximations}} 

Let $\domain_n \Cara \domain$ and let $\confmap_n: \domain_n \to \UnitD$ and $\confmap: \domain \to \UnitD$ be conformal maps with $\confmap_n^{-1} \to \confmap^{-1}$ uniformly over compact subsets of $\UnitD$. Suppose in addition that $\C \setminus \domain_n$ are uniformly locally connected, as defined in~\cite[Section~2.2]{Pommerenke-boundary_behaviour_of_conformal_maps} (e.g., this occurs if $\C \setminus \domain$ is locally connected and $\domain_n$ is, for each $n$, the domain bounded by a simple loop on $\frac{1}{n} \Z^2$ that stays inside $\domain$ and encloses a maximal number of squares).
Then, the conformal maps $\confmap_n^{-1}, \confmap^{-1}$ can be continuously extended to $\overline{\UnitD}$, and also the convergence $\confmap_n^{-1} \to \confmap^{-1}$ is uniform over $\overline{\UnitD}$~\cite[Theorem~2.1 and Corollary~2.4]{Pommerenke-boundary_behaviour_of_conformal_maps}. In particular $\confmap^{-1} $ and $
\confmap_n^{-1}$ are also continuous as functions $X (\overline{ \UnitD }) \to X(\C)$.

\begin{prop}
\label{prop: easy special case}
Let $ \gamma^{(n)}_\UnitD  \to   \gamma_\UnitD$ weakly in $X (\overline{ \UnitD })$, and suppose that $\confmap_n^{-1} \to \confmap^{-1}$ uniformly over $\overline{\UnitD}$. Then,
$\confmap_n^{-1} ( \gamma^{(n)}_\UnitD ) \stackrel{n \to \infty}{\longrightarrow} \confmap^{-1} ( \gamma_\UnitD )$
weakly in $X(\C)$.
\end{prop}

\begin{proof} 
It is sufficient to show weak convergence with a bounded Lipschitz-continuous test function $f:  X(\C ) \to \R$. Compute
\begin{align*}
\vert \EX^{(n)} & [f(\confmap_n^{-1} ( \gamma^{(n)}_\UnitD ) )] - \EX [f( \confmap^{-1} (\gamma_\UnitD ) )]  \vert \\
& \le \vert \EX^{(n)} [f(\confmap_n^{-1} ( \gamma^{(n)}_\UnitD ) ) %] - \EX^{(n)} [
- f( \confmap^{-1} (\gamma^{(n)}_\UnitD  ) )]  \vert + \vert \EX^{(n)} [f(\confmap^{-1} ( \gamma^{(n)}_\UnitD ) )] - \EX [f( \confmap^{-1} (\gamma_\UnitD ) )]  \vert.
\end{align*}
The second term on the right-hand side becomes arbitrarily small as $n \to \infty $ due to the weak convergence $\gamma^{(n)}_\UnitD \stackrel{n \to \infty}{\longrightarrow} \gamma_\UnitD  $. The first term becomes arbitrarily small by the Lipschitzness of $f$ and the fact that $\confmap_n^{-1} \to \confmap^{-1}$ uniformly over $\overline{\UnitD}$, and hence also $\confmap_n^{-1} \to \confmap^{-1}$ uniformly as maps $X(\overline{\UnitD}) \to X(\C ) $.
\end{proof}

\subsubsection{\textbf{No boundary visits}}

The main application is also readily handled if the boundary visits of $\gamma^{(n)}_\UnitD $ and $\gamma_\UnitD$ can be restricted (\textit{a priori} or \textit{a posteriori}) to only boundary segments with enough boundary regularity.

\begin{prop}
\label{prop:easy 2}
Let $ \gamma^{(n)}_\UnitD  \to   \gamma_\UnitD$ weakly in $X (\overline{ \UnitD })$. Suppose that $\gamma_\UnitD$ is, almost surely, a chordal curve from $-1$ to $1$ hitting no other boundary points, and that the maps $\confmap_n^{-1} $ and $ \confmap^{-1}$ and their uniform convergence can be extended to $\overline{V}_{\pm 1}$, where ${V}_{\pm 1}$ are some neighbourhoods of $\pm 1$ in $\UnitD$. Then, $\confmap_n^{-1} ( \gamma^{(n)}_\UnitD ) \to \confmap^{-1} ( \gamma_\UnitD )$  weakly in $X(\C)$.
\end{prop}

A proof for this proposition is obtained, e.g., from the previous one by approximating $f$ with $gf$, where $g$ is a suitable continuous cut-off function that removes the (unlikely) cases when a curve comes close to $\bdry \UnitD \setminus \overline{V}_{\pm 1}$. We leave the details for the reader. Proposition~\ref{prop:easy 2} is often sufficient, e.g., for multiple SLE applications analogous to Figure~\ref{fig:schematic application} when $\kappa \leq 4$.
%, if $\domain_n$ are uniformly locally connected and $\gamma^{(n)}_1 \to \gamma_1$, then the conformal maps $\confmap_n^{-1}$ to the right-most connected components of $\domain_n \setminus \gamma^{(n)}_1$ have the above required boundary extensions. For $\kappa > 4$ as in Figure~\ref{fig:schematic application}(left), $\gamma_2$ \textit{does} make boundary visits outside of its end points, but for $\kappa \leq 4$, the above special case is often sufficient for multiple SLE applications.

\subsection{Warning examples}
\label{subsec:warning}

\begin{figure}
\includegraphics[ width=0.25\textwidth]{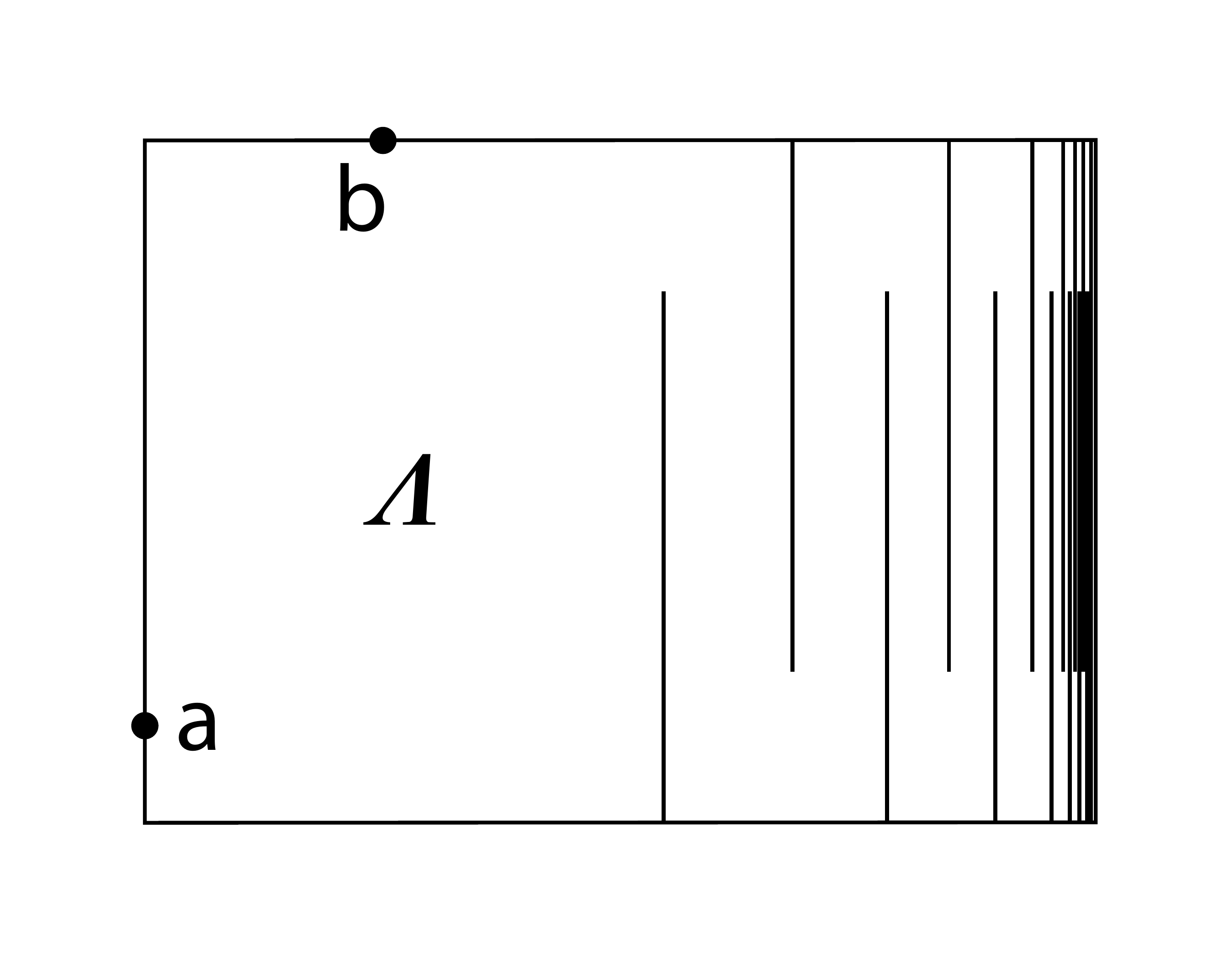} 
\caption{\label{fig: bad domain} An $\SLE(8)$ curve $\gamma_\UnitD$ in $(\UnitD; -1,1)$ cannot be mapped conformally to a chordal curve  $\gamma: [0, 1] \to \C$ in this domain $(\domain; a,b)$.}
\end{figure}

We start by remarking, perhaps trivially, that it is easy to give examples of (non-conformal) homeomorphisms $f_n : \domain_n \to \UnitD$ and $f : \domain \to \UnitD$, such that $f_n^{-1} \to f^{-1}$ uniformly over compact subsets of $\UnitD$, and suitable $\gamma^{(n)}_\UnitD \to \gamma_\UnitD  $ (even deterministic curves)  so that no convergence ``$f_n^{-1} ( \gamma^{(n)}_\UnitD )  \to f^{-1} ( \gamma_\UnitD)  $'' occurs. The conformal structure will indeed play a key role in the proofs.

As a first real warning, general Carath\'{e}odory converging domains $\domain_n \to \domain$ may contain ``deep fjords'', and it may occur that $\gamma^{(n)} \to \gamma$ and $\gamma^{(n)}_\UnitD \to \gamma_\UnitD$ weakly but $\gamma$ does not even stay in $\overline{\domain}$ (see also Figure~\ref{fig: bad boundary}), certainly preventing any relation ``$\gamma = \confmap^{-1} ( \gamma_\UnitD ) $''. The conditions on $\gamma^{(n)}$ from~\cite{KS} however rule this out, apart from the curve end points which we will handle separately in our main theorem.

Secondly, for rough domains $\domain$, making sense of a random curve ``$ \confmap^{-1} ( \gamma_\UnitD ) $'' in the first place is an issue (without any reference to weak convergence): $\confmap^{-1}$ may not have any natural extension to $\overline{\UnitD}$ and even if it has, e.g. by radial limits, the extension may be discontinuous. For instance, the $\SLE(8)$ in $(\UnitD; -1, 1)$ exists as a random Peano curve $\gamma_\UnitD \in X( \overline{ \UnitD } )$ and appears as a scaling limit of a lattice model~\cite{LSW-LERW_and_UST}, but if $\confmap^{-1} : (\UnitD; -1, 1) \to (\domain; a,b)$ is a conformal map to the domain in Figure~\ref{fig: bad domain}, then one can show that, almost surely, there exists no continuous curve $\gamma: [0, 1] \to \C$ such that $\gamma(t) = \confmap^{-1} ( \gamma_\UnitD (t) ) $ for all $t$ with $\gamma_\UnitD (t) \in \UnitD$.
 For $4 < \kappa < 8$, $\SLE (\kappa)$ type curves $\gamma_\UnitD$ visit the boundary on an uncountable set with no isolated points. Straightforwardly defining a curve ``$ \confmap^{-1} ( \gamma_\UnitD ) $'' would thus still require a strong control of $\confmap^{-1}$ on a large-cardinality random boundary set, cf.~\cite{GRS-SLE_in_sc_domains}.
 
Finally, we will have to gain such control uniformly over all $n$.
Informally speaking, the strategy will be to control the radial derivative of $\confmap^{-1}_n$ near $\gamma_\UnitD^{(n)}$ with high probability, uniformly over $n$. This is done by restricting large derivatives to preimages of ``deep fjords'', uniformly over the target domains, while travelling into deep fjords is excluded by~\cite{KS}. Combining with the weak convergence $(\gamma_\UnitD^{(n)}, \confmap^{-1}_n (\gamma_\UnitD^{(n)})) \to (\gamma_\UnitD, \gamma)$ given by precompactness, one then deduces that $\gamma = \confmap^{-1} ( \gamma_\UnitD ) $, in the sense of radial limits. The conclusions can be contrasted with~\cite{GRS-SLE_in_sc_domains} and some other results in ``continuous'' SLE theory, see Section~\ref{sec: applications}.

\section{The main results}

\subsection{Setup}
\label{subsec: setup}

The following setup and notation applies throughout this section. Let $(\domain_n; a_n, b_n)$, $n \in \N$, be simply-connected planar domains with two distinct marked prime ends with radial limits, and let $\confmapBdry_n: \domain_n \to \UnitD$ be conformal maps such that $\confmapBdry_n^{-1} (-1) = a_n$ and $\confmapBdry_n (1) = b_n$ (in the sense of the prime end bijection). A \textit{random chordal Loewner curve model} on $(\domain_n; a_n, b_n)$ is, formally, a probability space $(\PR^{(n)}, \Omega^{(n)}, \mathscr{F}^{(n)})$ with random variables $\gamma^{(n)} \in X(\C)$, $\gamma_\UnitD^{(n)} \in X(\overline{\UnitD})$, and $W^{(n)} \in \ctsfcns$ such that, almost surely, $\gamma_\UnitD^{(n)} $ traverses from $-1$ to $1$, $\gamma^{(n)}$ satisfies $\gamma^{(n)} = \confmapBdry_n^{-1} ( \gamma_\UnitD^{(n)}  )$ (extending $\confmapBdry_n^{-1}$ by radial limits), and $W^{(n)}$ is the Loewner transform of $\gamma_\UnitD^{(n)}$.
We equip $\mathscr{F}^{(n)}$ with the right-continuous filtration $(\mathscr{F}^{(n)}_t)_{t \ge 0}$ of the driving functions $W^{(n)}_{t} \in \ctsfcns$; stopping times refer to this filtration. While a stopping time $\tau$ is thus defined on the interval $[0, \infty)$, we will use abusive notations such as $\gamma^{(n)}([0, \tau])$ for the corresponding segments on curves parametrized by times in $[0,1]$.

\subsection{Kemppanen and Smirnov's theorem}

Let us first review Kemppainen and Smirnov's results.
We start with the hypothesis, formulated in terms of certain crossing conditions. For equivalent hypotheses, see~\cite[Section~2.1.4]{KS}.

For $0 < r < R$ and $z \in \C$, we denote
\begin{align*}
B(z, r) = \set{w \in \C \ : \ \vert w -z \vert < r}, \quad
S(z, r) = \bdry B(z, r), \quad
\quad \text{and} \quad
A(z, r, R) = B(z, R) \setminus \overline{B(z, r)}.
\end{align*}
We say that $(Q; S_0, S_1, S_2, S_3)$  is a \textit{boundary quad} (resp. $A(z, r, R)$ is a \textit{boundary annulus}) of a simply-connected domain $\domain \subset \C$ if if $Q \subset \domain$, $S_1, S_3 \subset \bdry \domain$, and $S_0, S_2 \subset \domain$ (resp. if $ B(z, r) \cap \bdry \domain \ne \emptyset$). We say that a boundary quad (resp. a connected component of a boundary annulus in $\domain$) \textit{disconnects} two prime ends $a$ and $b$ in $\domain$ if it disconnects some neighbourhood of $a$ in $\domain$ from some neighbourhood of $b$ in $\domain$. Finally, for a chordal curve $\gamma$ from $a$ to $b$ in $\domain$, a crossing of a boundary quad $Q$ between $S_0$ and $S_2$ by  is \textit{unforced} if $Q$ does not disconnect $a$ from $b$ (resp. a crossing of a boundary annulus $A(z, r, R)$ between $\bdry B(z, r)$ and $\bdry B (z, R)$ is \textit{unforced} if it occurs in a connected component of $A(z, r, R)$ that does not disconnect $a$ from $b$). We can now formulate the hypothesis (see Figure~\ref{fig: Kemppainen crossing condns}).

\begin{figure}
\includegraphics[width=0.3\textwidth]{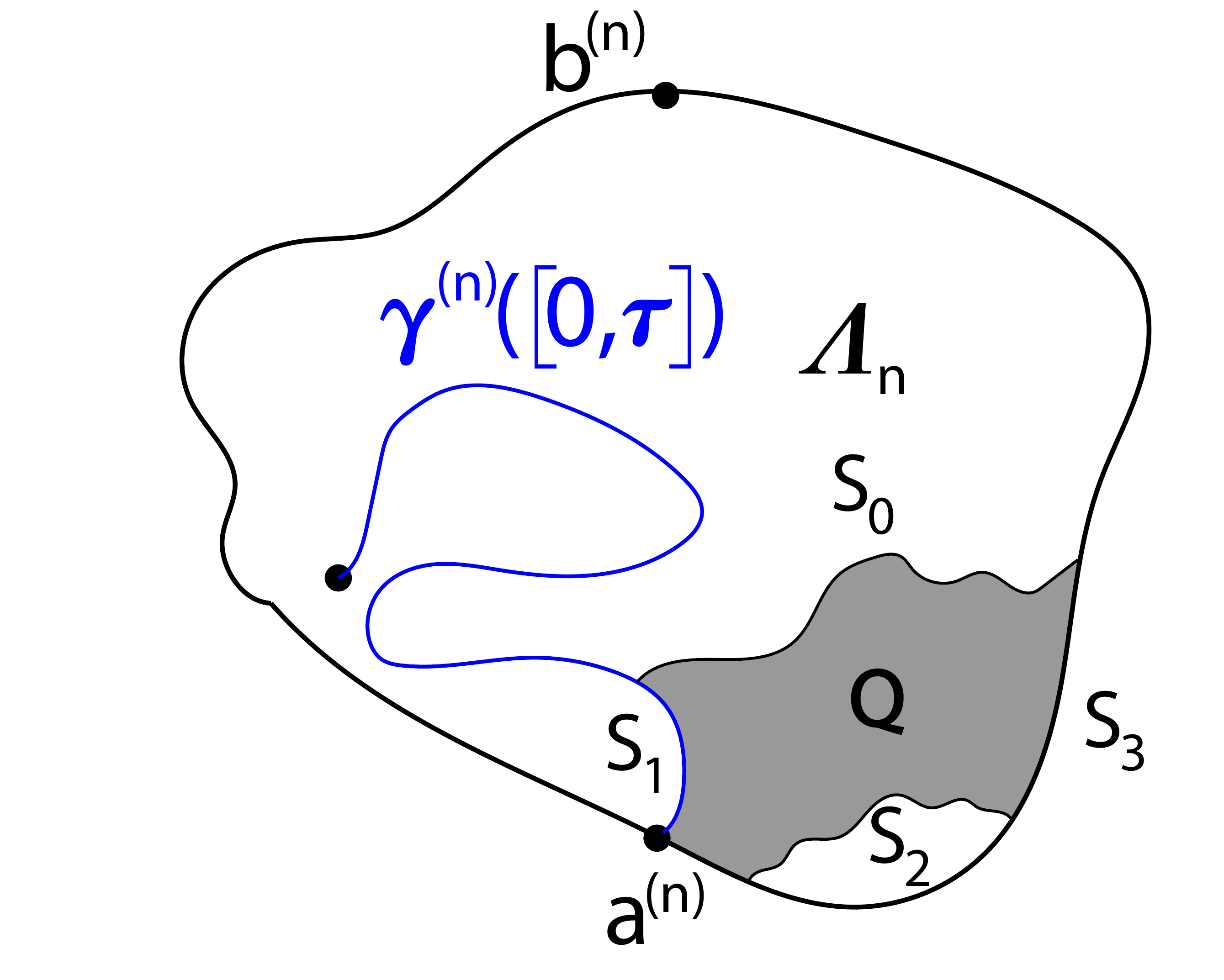}
\qquad
\includegraphics[width=0.3\textwidth]{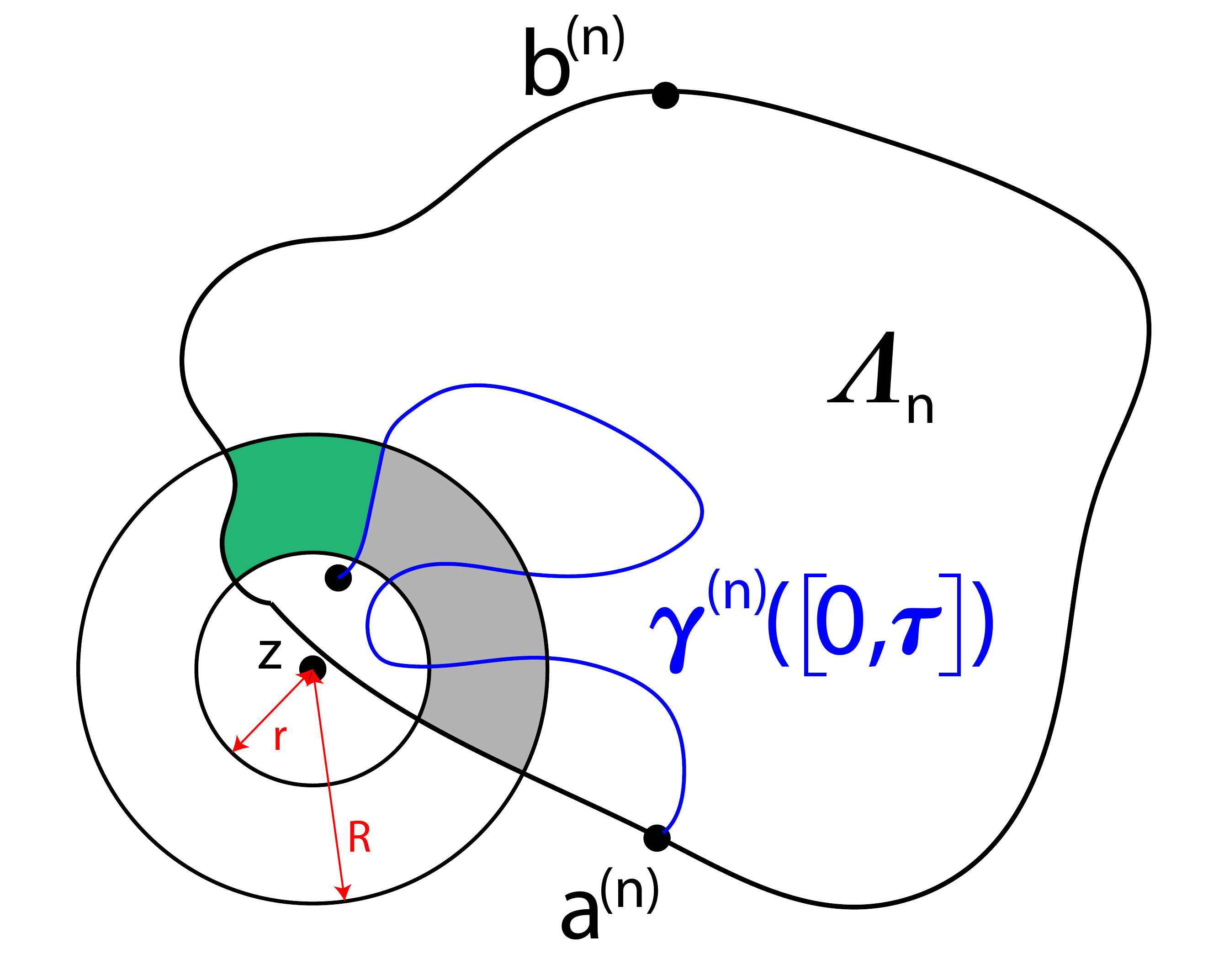}
\caption{
\label{fig: Kemppainen crossing condns} Illustrations of the crossing conditions of~\cite{KS}. Condition (C) depicted in the left figure requires that it is unlikely that the remainder $\gamma^{(n)}([\tau, 1])$ of the curve $\gamma^{(n)}$ crosses a quadrilateral $Q$ of large conformal modulus $m(Q)$ on the boundary of $\domain_n \setminus \gamma^{(n)}([0, \tau])$. The equivalent Condition (G) depicted on the right requires that an unforced crossing of a boundary annulus $A(z, r, R)$ of large ratio $R/r$ is unlikely. The connected components of $A(z, r, R) \cap \domain_n \setminus \gamma^{(n)}([0, \tau])$ where a crossing is unforced are depicted in gray; the remainder curve $\gamma^{(n)}([\tau, 1])$ is topologically forced to cross the component depicted in green.
}
\end{figure}
 
Condition (C): We say that the measures $\PR^{(n)}$ \emph{satisfy condition (C)} if for all $\eps > 0$  there exists $M > 0$, independent of $n$, such that the following holds for all stopping times $\tau < \infty$: for any quad $Q$ with $m(Q) \ge M$ on the boundary of $\domain_n \setminus \gamma^{(n)}([0, \tau])$, we have
\begin{align*}
\PR^{(n)} [\gamma^{(n)}([\tau, 1]) \text{ makes an unforced crossing of } Q
%\text{ unforced in } (\domain_n \setminus \gamma^{(n)}([0, \tau]), \gamma^{(n)}(\tau), b_n  )
\; \vert \; \mathscr{F}^{(n)}_\tau ] \le \eps, \qquad \text{almost surely}.
\end{align*}

Condition (G): We say that the measures $\PR^{(n)}$ \emph{satisfy condition (G)} if for all $\eps > 0$  there exists $M > 0$, independent of $n$, such that the following holds for all stopping times $\tau < \infty$: for any annulus $A(z, r, R)$ with $R/r \ge M$ on the boundary of $\domain_n \setminus \gamma^{(n)}([0, \tau])$, we have
\begin{align*}
\PR^{(n)} [\gamma^{(n)}([\tau, 1]) \text{ makes an unforced crossing of }  A(z, r, R)
%\text{ unforced in } (\domain_n \setminus \gamma^{(n)}([0, \tau]), \gamma^{(n)}(\tau), b_n  )
\; \vert \; \mathscr{F}^{(n)}_\tau ] \le \eps, \qquad \text{almost surely}.
\end{align*}

We also point out~\cite[Remark~2.9]{KS} here: if the curves $\gamma^{(n)}$ live on finite graphs, it suffices to check these conditions for stopping times in the sparser, discrete filtration $(\mathscr{F}^{(n)}_{\tau_i})_{i \in \N}$, where $\tau_i$ is the hitting time of the $i$:th vertex on the curve $\gamma^{(n)}$. We are now ready to state the theorem.
 
 \begin{thm}
 \label{thm: Kemppainen's precompactness} \cite[Theorem~1.5 and Corollary~1.7]{KS}
In the setup of Section~\ref{subsec: setup}, suppose that the measures $\PR^{(n)}$ satisfy the equivalent conditions (C) and (G).
Then, the laws under $\PR^{(n)}$ of the pairs $(\gamma^{(n)}_\UnitD, W^{(n)})$ on  $X(\overline{\UnitD}) \times \ctsfcns$ are precompact. Furthermore, for any subsequential limit pair, the curve and driving function are, almost surely, Loewner transforms of each other.
 \end{thm}
 
Note that a direct consequence of the above is that if a weak convergence takes place for either $\gamma^{(n)}_\UnitD$ or $ W^{(n)}$, then it also takes place for their joint law.

A careful reader may observe two minor differences to the formulation in~\cite{KS} which are however essentially of presentational nature. First, the assumptions in Section~\ref{subsec: setup} were given following the equivalent form from~\cite[Section~1.4]{KS}, rather than~\cite[Theorem~1.5 and Corollary~1.7]{KS} directly. Second,~\cite[Theorem~1.5 and Corollary~1.7]{KS} were not formulated in product topology. 
Tightness (and hence precompactness) in the product topology however follows directly from the tightness of $\gamma^{(n)}_\UnitD$ and $ W^{(n)}$ individually. 
As regards the the Loewner transform property for the limit in the product topology, the proof of~\cite[Corollary~1.7]{KS} is based on finding, for any $m> 0$, a set of Loewner transformable curves $K_{m}$ that carries a $\PR^{(n)}$-probability at least $1-1/m$, is compact in the topology of both $\gamma^{(n)}_\UnitD$ and $ W^{(n)}$, and on which the two are continuous functions of each other. The limiting objects are thus also supported on $\cup_m K_m$ and the Loewner transform property follows, e.g., by showing with Tietze's extension theorem that $ \EX [ |f($Loewner$(W)) - f(\gamma_\UnitD)| ] = 0$, where $\EX$, $W$, and $\gamma_\UnitD$ are the limit objects $f: X(\overline{\UnitD}) \to \R$ is any bounded  continuous test function.

\subsection{The contribution of the present note}
\label{subsec: Improving Kemppainen's precompactness}
\label{subsubsec: close approximations}

We will need the following notion to formulate our main theorem. Suppose that $(\domain_n; a_n) \to (\domain; a)$ in Carath\'{e}odory sense with respect to $u$ and that radial limits exist at the prime ends $a_n$ and $a$; let us for lighter notation identify $a_n$ and $a$ with the corresponding radial limit points in $\C$. For $r \in ( 0, d (a, u))$, let $S_r$ denote the connected component of $\bdry B(a, r) \cap \domain $ that disconnects the point $a$ from $u$ and is closest to $a$ (such exists by Lemma~\ref{lem: disconnecting arcs} and approximation by radial limits). We say that $a_n$ are \textit{close approximations} of $a$ (see Figure~\ref{fig: bad boundary}) if
$a_n \to a$ as points in $\C$ and for any $r \in ( 0, d (a, u))$ and fixed $w_r \in S_r$, $a_n$ is connected to $w_r$ inside $\domain_n \cap B(a, r)$ for all $n$ a large enough, $n > n_0(r, w_r)$. (By Carath\'{e}odory convergence, it suffices to verify this condition for a given choice of $w_r$.)

Our main theorem concerns the curves $\gamma^{(n)}$ in the original domains, adding the following assumptions to our setup described in Section~\ref{subsec: setup}: 
\begin{itemize}
\item (bounded domains) there exists $M>0$ such that $\domain \subset B(0, M)$ and $\Area (\domain_n ) \leq M$ for all $n$;
\item (Carath\'{e}odory convergence) $ (\domain_n; a_n, b_n) \Cara (\domain; a, b)$, and the conformal maps $\confmapBdry_n: (\domain_n; a_n, b_n) \to (\UnitD; -1, 1)$ and $\confmapBdry: (\domain; a, b) \to (\UnitD; -1, 1)$ normalized at boundary points are chosen so that $\confmapBdry^{-1}_n \to \confmapBdry^{-1}$ uniformly over compact subsets of $\UnitD$; and
\item (close approximations) the prime ends $a_{n}$ and $b_{n}$ in $\domain_n$ have radial limits and are close approximations of the prime ends $a $ and $b$ in $\domain$ with radial limits.
\end{itemize} 

\begin{figure}
\includegraphics[ width=0.25\textwidth]{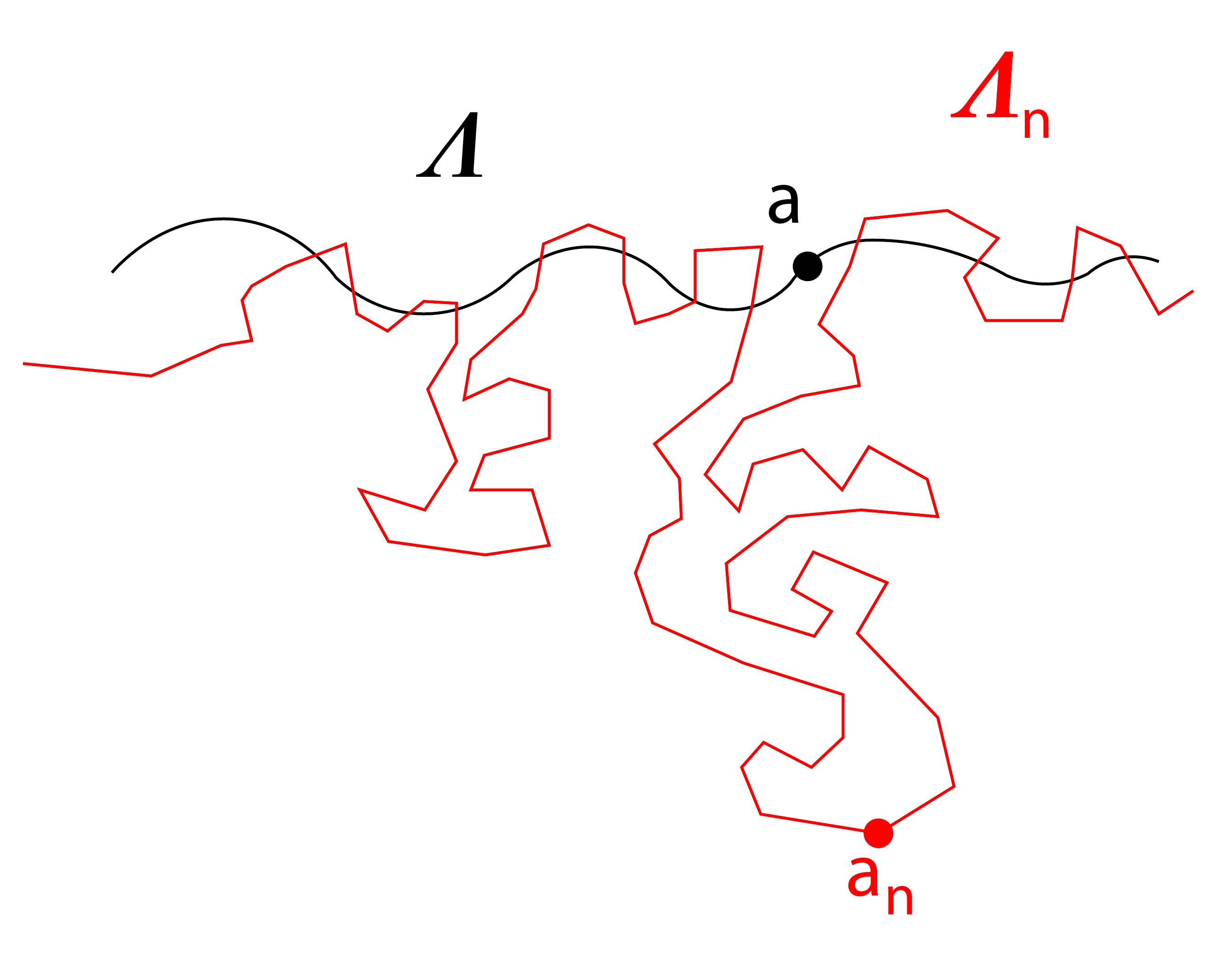} \qquad
\includegraphics[ width=0.25\textwidth]{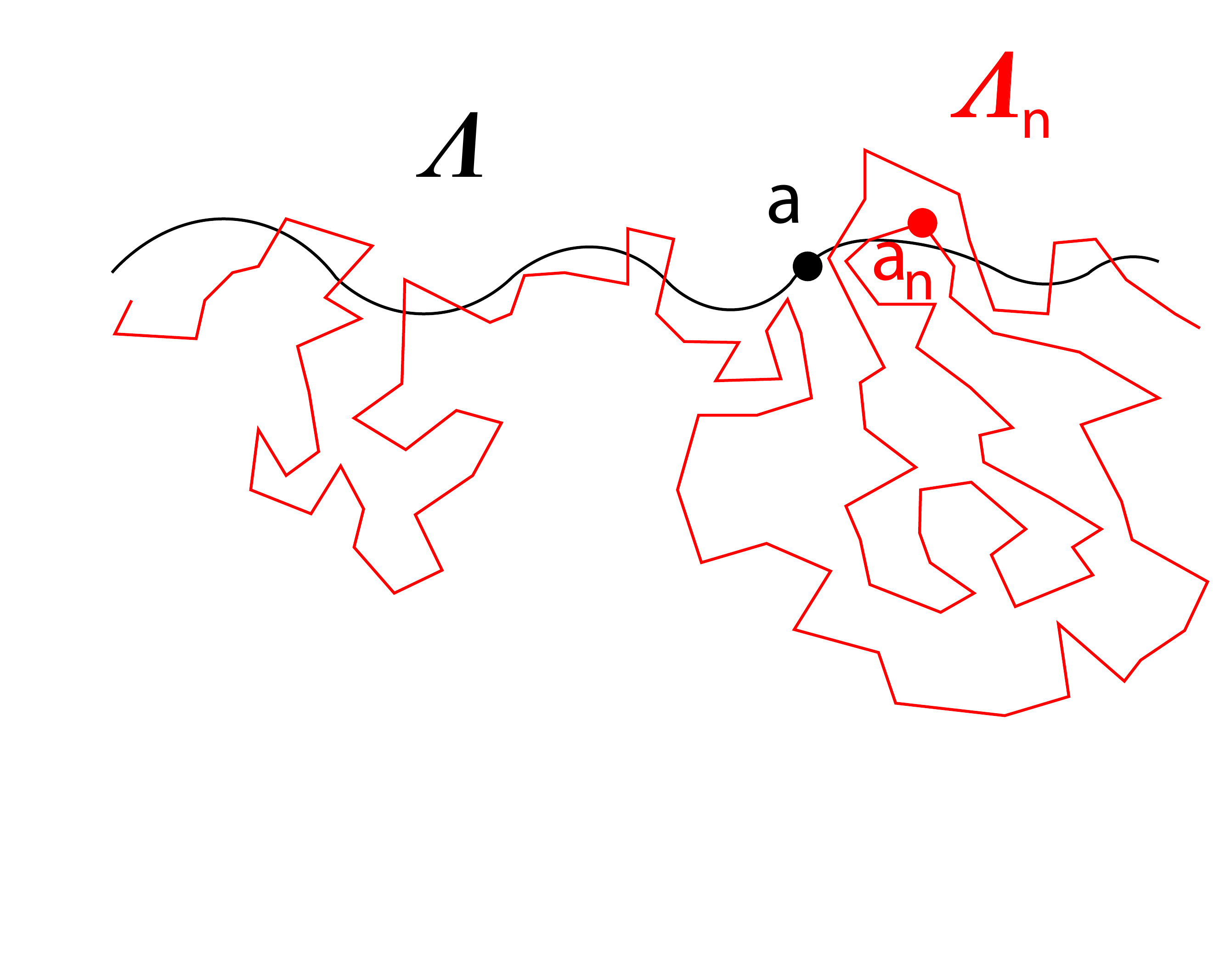}
\caption{\label{fig: bad boundary} Schematic illustrations of undesired behaviour of boundary approximations $(\domain_n; a_n) \to (\domain; a)$. Left: $a_n \not \to a$ as points in $\C$. Right: $a_n \to a$ but $a_n$ are not close approximations.}
\end{figure}

We are now ready to formulate our main theorem. The precompactness part, apart from the behaviour at the marked boundary points, originates from~\cite{AB-regularity_and_dim_bounds_for_random_curves, KS}.

\begin{thm}
\label{thm: extension of Kemppainen's precompactness}
Under the setup given above and in Section~\ref{subsec: setup}, suppose that the measures $\PR^{(n)}$ satisfy the equivalent conditions (C) and (G). Then, the laws of $(\gamma_\UnitD^{(n)}, \gamma^{(n)}) \in X(\overline{\UnitD}) \times X (\C)$ under $\PR^{(n)}$
 are precompact. Furthermore, any subsequential limit $(\gamma_\UnitD, \gamma)$ almost surely satisfies $\gamma = \confmapBdry^{-1} ( \gamma_{\UnitD } )$, where $\confmapBdry^{-1} $ is extended by radial limits.
\end{thm}

Before the proof, let us make some remarks. First, combining Theorems~\ref{thm: Kemppainen's precompactness} and~\ref{thm: extension of Kemppainen's precompactness}, one straightforwardly obtains an analogous result for the triplets $(W^{(n)}, \gamma_\UnitD^{(n)}, \gamma^{(n)})$. Consequently, if weak convergence takes place for either  $ W^{(n)}$, $\gamma^{(n)}_\UnitD$, or $\gamma^{(n)}$, then it also takes place for their joint law, and the limit objects are Loewner/conformal transforms of each other, as detailed in Theorems~\ref{thm: Kemppainen's precompactness} and~\ref{thm: extension of Kemppainen's precompactness}, in particular satisfying the informal ``commutative diagram'' in the Introduction.

Second, as discussed in Section~\ref{subsec:warning}, it is not immediate, but a part of the result that, almost surely, $\confmapBdry^{-1} $ (as extended by radial limits) is defined an all points of $\gamma_{\UnitD }$ and the function $\confmapBdry^{-1} \circ \gamma_{\UnitD }: [0, 1] \to \C$ is continuous (i.e., a curve). Furthermore, this curve is a measurable $X(\C)$ valued random variable in the sigma algebra of $\gamma_{\UnitD }$; this follows from Proposition~\ref{prop: weak limit commutation property}.

It is also an interesting question what regularity properties can be proven for the limit curve $\gamma$ above. The proof of the tightness below boils down to applying~\cite{AB-regularity_and_dim_bounds_for_random_curves} to interior segments of $\gamma^{(n)}$, from the time of first exiting a small neighbourhood of $a_n$ to the time of first hitting a small neighbourhood of $b_n$. The H\"{o}lder regularity and dimension bounds of~\cite{AB-regularity_and_dim_bounds_for_random_curves} thus follow for analogous interior segments of $\gamma$. However, these regularity properties \emph{generally do not hold} for the full curve $\gamma$ without additional boundary regularity assumptions for the domains; an explicit counterexample can be worked out from the idea depicted in Figure~\ref{fig: Peano fjord}.

\begin{figure}
\includegraphics[width=0.4\textwidth]{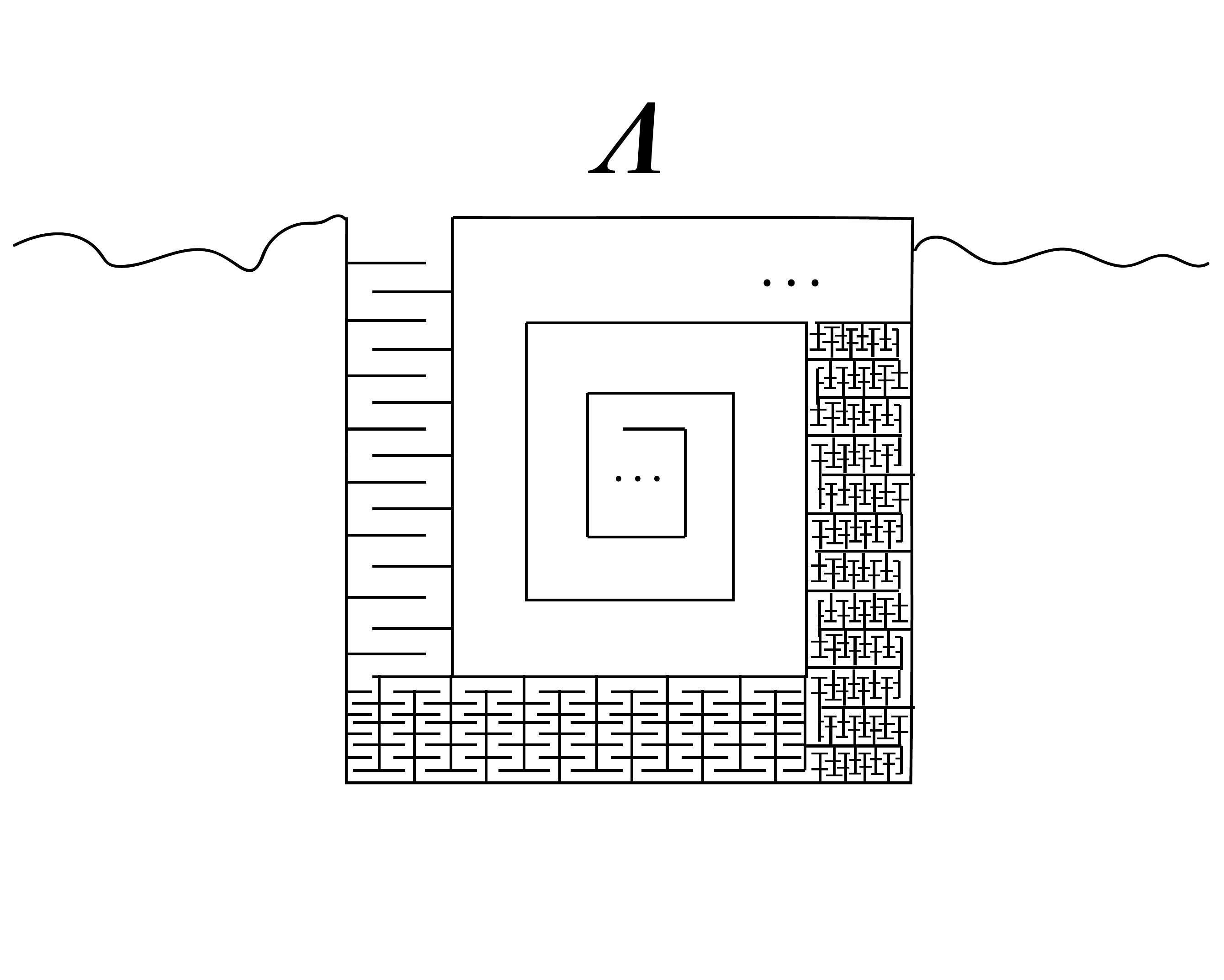}
\caption{
\label{fig: Peano fjord}
A schematic illustration of a counterexample about the regularity of the limit curves $\gamma$ in Theorem~\ref{thm: extension of Kemppainen's precompactness}. If the limiting domain $\domain$ that has on its boundary an ``infinite spiral of corridors with multiple layers of obstacles in each corridor'', then the prime end at the bottom of the spiral has radial limits. However, one can show that if the number of obstacle layers grows fast enough towards the bottom of the spiral (in the picture, the number of obstacle layers is drawn for simplicity equal to the number of the corridor), then any curve reaching the bottom of the spiral and not crossing $\bdry \domain$ must be of upper box dimension two. 
}
\end{figure}

\subsection{Proof of Theorem~\ref{thm: extension of Kemppainen's precompactness}}

\begin{proof}[Proof of precompactness in Theorem~\ref{thm: extension of Kemppainen's precompactness}]
It suffices to show that the laws of $\gamma^{(n)} \in X(\C) $ under $\PR^{(n)}$ are precompact. Indeed (using Prokhorov's theorem twice) then $\gamma^{(n)}$ are tight and, combining with Theorem~\ref{thm: Kemppainen's precompactness}, also $(\gamma_\UnitD^{(n)}, \gamma^{(n)})$ are tight and hence precompact.

Fix $\varrho > 0$ and denote by $S^{(a)}$ the innermost arc of the circle $S(a, \varrho )$ in $\domain$ that disconnects the boundary point $a$ from $u$. Let $w^{(a)} \in S^{(a)}$ be the corresponding reference point, and let $S^{(a)}_n$ be the arc of the circle $S(a, \varrho )$ in $\domain_n$ that contains $w^{(a)}$ (such an arc exists for all large enough $n > n_0 (\varrho)$). Finally, define the similar objects for $b$. Denote by $\gamma^{(n)}_\varrho$ the segment $\gamma^{(n)}([\tau_{S_n^{(a)}}, \tau_{S_n^{(b)}} ])$ of the curve $\gamma^{(n)}$, where $\tau_{S_n^{(a)}}$ and $ \tau_{S_n^{(b)}}$ are the first hitting times of $S_n^{(a)}$ and $S_n^{(b)}$, respectively. For any fixed $\varrho > 0$, the curves $\gamma^{(n)}_\varrho$ are precompact in $X(\C)$ by~\cite[Corollary~1.8]{KS}.

Given $\tilde{\epsilon} > 0$, fix $\chi = \chi(\tilde{\epsilon})$ through property (G) so that an unforced crossing of any fixed boundary annulus $A(z, r, \chi r)$ occurs with $\mathbb{P}^{(n)}$-probability at most $\tilde{\eps}/2$, for any $n$.
We next claim that for any $\varrho, \tilde{\epsilon}$ and $\chi$, there exist  $n_0 = n_0( \varrho, \tilde{\epsilon}, \chi )$ such that for all $n > n_0$,
\begin{align}
\label{eq: lip estimate}
\PR^{(n)} [d (\gamma^{(n)}, \gamma^{(n)}_\varrho) \ge 2 \chi \varrho ] \le \tilde{\eps }.
\end{align}
Indeed, choose $n_0$ through the close approximation property so that for $n > n_0$, $S_n^{(a)}$ exists, $a_n \in B(a, \varrho)$ and $a_n$ is connected to  $S_n^{(a)}$ inside $B(a, \varrho) \cap \domain_n$, and the similar properties hold for $b$. Thus, any crossing of $A(a, \varrho, \chi \varrho)$ before $\tau_{S_n^{(a)}}$ is unforced (for any $\chi > 1$), and likewise is any crossing of $A(b, \varrho, \chi \varrho)$ after $\tau_{S_n^{(b)}}$. If such crossings do not occur, truncating $\gamma^{(n)}$ to $ \gamma^{(n)}_\varrho$ perturbs the curve at most by $2 \chi \varrho$. The choice of $ \chi(\tilde{\epsilon})$ then implies~\eqref{eq: lip estimate}.

The proof can now be finished with a general argument about complete metric spaces. To explain it, recall first that weak convergence of Borel probability measures on metric spaces can be metrized by the L\'{e}vy-Prokhorov metric $d_{LP}$, defined by
\begin{align*}
d_{LP} (\mu, \nu) = \inf \{ \epsilon > 0 : \mu(A) \leq \nu (A^\epsilon) + \epsilon \text{ and } \nu(A) \leq \mu (A^\epsilon) + \epsilon \text{ for all Borel sets } A \},
\end{align*}
when $A^\epsilon$ denotes the open $\epsilon$-neighbourhood of $A$. For instance,~\eqref{eq: lip estimate} directly gives
\begin{align}
\label{eq:expl LP dist}
d_{LP} (\gamma^{(n)}, \gamma^{(n)}_\varrho) \leq \max \{ 2 \chi \varrho, \tilde{\eps } \},
\end{align}
where we (slightly abusively) denoted the distance of laws by the corresponding random variables. The space of probability measures with the L\'{e}vy-Prokhorov metric is in addition complete if the underlying metric space is (as it is here). The general argument (in the notation of a complete L\'{e}vy-Prokhorov space) is the following: a sequence of probability measures $(\mu_n)_{n \in \N}$ is precompact if there exists an array of probability measures  $(\nu_{n,m})_{n, m \in \N}$ such that that $d_{LP} (\mu_n, \nu_{n,m}) \leq 1/m$ for all $n, m$ and that the sequence $(\nu_{n, m})_{n \in \N}$ is precompact for every $m$. Indeed, by diagonal extraction, there is a subsequence $n_k$ such that $(\nu_{n_k, m})_{k \in \N}$ converges weakly for all $m$ and that
\begin{align*}
d_{LP} (\nu_{n_{k},m}, \nu_{n_{k'},m}) < 1/k
\end{align*}
for all $m \leq k \leq k'$. Then,
\begin{align*}
d_{LP} (\mu_{n_k}, \mu_{n_{k'}}) \leq d_{LP} (\mu_{n_k}, \nu_{n_k,m})  + d_{LP} (\nu_{n_{k},m}, \nu_{n_{k'},m})  + d_{LP} ( \nu_{n_{k'},m}, \mu_{n_{k'}}) 
\leq
1/m + 1/k + 1/m
\end{align*}
for all $m \leq k \leq k'$, and choosing $m=k$, we observe that the sequence $(\mu_{n_k})_{k \in \N}$ is Cauchy.

To apply the general argument, let $\mu_n$ be the law of $\gamma^{(n)}$. Given $m$, set $\tilde{\epsilon}=1/m$, then $\chi = \chi(\tilde{\epsilon})$ as above, then $\varrho $ so that $2 \chi \varrho \leq 1/m$, and finally $n_0 = n_0( \varrho, \tilde{\epsilon}, \chi )$ as above. For $n \leq n_0$ let $\nu_{n, m} = \mu_n$ and for $n > n_0$ let $\nu_{n, m}$ be the law of $\gamma^{(n)}_\varrho$. By~\eqref{eq:expl LP dist}, $d_{LP} (\mu_n, \nu_{n,m}) \leq 1/m$ and, as noticed above, the laws $(\nu_{n, m})_{n \in \N}$ of $\gamma^{(n)}_{\varrho(m)}$ are precompact for all $m$. The claim now follows.
\end{proof}

For the proof of the conformal property of Theorem~\ref{thm: extension of Kemppainen's precompactness}, 
the Riemann maps $\confmap_n: \domain_n \to \UnitD$ (i.e. normalized at a fixed common point $u \in \domain_n$) turn out technically nicer; recall that $\confmap_n^{-1}$ converge to the inverse Riemann map $\confmap^{-1}$ of $\domain$ (normalized at $u$). Denote
\begin{align*}
\tilde{\gamma}^{(n)}_{\UnitD} = \confmap_n (\gamma^{(n)}) = \confmap_n \circ \confmapBdry_n^{-1} ( \gamma_\UnitD^{(n)}  );
\end{align*}
note that by the ``easy special case'' of Proposition~\ref{prop: easy special case}, $\tilde{\gamma}^{(n)}_{\UnitD} \to \tilde{\gamma}_{\UnitD}$ weakly if and only if ${\gamma}^{(n)}_{\UnitD} \to {\gamma}_{\UnitD}$ weakly, where $\tilde{\gamma}_{\UnitD} = \confmap \circ \confmapBdry^{-1} ( \gamma_\UnitD  ) $ (since $\confmap_n \circ \confmapBdry_n^{-1}$ and $\confmap \circ \confmapBdry^{-1}$  are M\"{o}bius maps $\overline{\UnitD} \to \overline{\UnitD}$). It is thus equivalent to show that $\gamma = \confmap^{-1} (\tilde{\gamma}_{\UnitD}) $, for any subsequential limit.
Recall also the notation for radial projections from Section~\ref{subsubsec: Cara convergence}. We first prove the following.

\begin{prop}
\label{prop: weak limit commutation property}
Let $(\gamma_\UnitD , \gamma )$ be a limiting pair from Theorem~\ref{thm: extension of Kemppainen's precompactness}, and $\tilde{\gamma}_\UnitD = \confmap \circ \confmapBdry^{-1} (\gamma_\UnitD)$.
For any sequence $\epsilon_{j} \downarrow 0$, there exists a subsequence $(\epsilon_{j_m})_{m \in \N}$ such that, almost surely, $  \confmap^{-1} \circ P_{\epsilon_{j_m}} (\tilde{\gamma}_\UnitD ) \to \gamma$ in $X(\C)$ as $m \uparrow \infty$.
\end{prop}

\begin{proof}
For notational simplicity let us suppress a subsequence notation and assume that $(\gamma_\UnitD^{(n)}, \gamma^{(n)}) \to (\gamma_\UnitD , \gamma )$ weakly. 
By a theorem of Skorokhod, there exists a coupling $\PR$ such that $(\tilde{\gamma}_\UnitD^{(n)}, \gamma^{(n)}) \to (\tilde{\gamma}_\UnitD , \gamma )$ almost surely.
 In this coupling, for any $\epsilon_j =: \epsilon \in (0,1)$ and any $n \in \N$ on the right-hand side
\begin{align}
\nonumber
d (\gamma, \confmap^{-1} \circ P_\epsilon (\tilde{\gamma}_\UnitD ) ) 
 \leq & 
d(\gamma, \gamma^{(n)} ) + d(\gamma^{(n)}, \confmap_n^{-1} \circ P_\epsilon (\tilde{\gamma}^{(n)}_\UnitD ) ) \\
\label{eq:four terms}
&
+ d( \confmap_n^{-1} \circ P_\epsilon (\tilde{\gamma}^{(n)}_\UnitD ), \confmap^{-1} \circ P_\epsilon (\tilde{\gamma}^{(n)}_\UnitD ) )
+ d( \confmap^{-1} \circ P_\epsilon (\tilde{\gamma}^{(n)}_\UnitD ), \confmap^{-1} \circ P_\epsilon (\tilde{\gamma}_\UnitD ) ).
\end{align}
The strategy of the proof is to use~\eqref{eq:four terms} to show that for every $m \in \N$ there exist $ \epsilon_{j_m}$ such that
\begin{align}
\label{eq:BC trick}
\PR \left[  d (\gamma, \confmap^{-1} \circ P_{\epsilon_{j_m}} (\tilde{\gamma}_{\UnitD} ) ) \geq  1/2^m \right] \leq 1/2^m.
\end{align}
Indeed, by~\eqref{eq:BC trick} and the Borel--Cantelli lemma, the events $\{ d (\gamma, \confmap^{-1} \circ P_{\epsilon_{j_m}} (\tilde{\gamma}_\UnitD ) ) \geq  1/2^m \}$ almost surely occur for only finitely many $m$, and the claim is proven.

We now analyze the terms on the right-hand side of~\eqref{eq:four terms} separately. For the first one, almost sure convergence implies convergence in probability, so
\begin{align}
\label{eq:term 1}
\PR & [ d(\gamma, \gamma^{(n)} ) \geq 1/2^{m+2} ] \leq 1/2^{m+2} \quad
\text{for all } n > n_0^{(1)},
\end{align}
where $n_0^{(1)}$ only depends on the coupling $\PR$.
For the third one, pick $n_0^{(3)} (\epsilon)$ by the convergence $\domain_n \Cara \domain$ so that $| \confmap_n^{-1} (\cdot) - \confmap^{-1} (\cdot) | < 1/2^{m+2}$ on $\overline{B(0, 1-\eps)}$, for all $n > n_0^{(3)} (\epsilon)$. Thus, deterministically,
\begin{align}
\label{eq:term 2}
d( \confmap_n^{-1} \circ P_\epsilon (\tilde{\gamma}^{(n)}_\UnitD ), \confmap^{-1} \circ P_\epsilon (\tilde{\gamma}^{(n)}_\UnitD ) ) <  1/2^{m+2} \qquad \text{for all } n > n_0^{(3)} (\epsilon).
\end{align}
For the fourth term, note first that $P_\epsilon$ satisfies $|P_\epsilon(z) - P_\epsilon (z') | \leq |z-z'|$ for all $z, z' \in \overline{\UnitD}$. Denoting $D_\epsilon = \sup_{z \in B(0, 1-\epsilon)} | ( \confmap^{-1} )' (z)|$ one then deduces that, deterministically,
\begin{align*}
d( \confmap^{-1} \circ P_\epsilon (\tilde{\gamma}^{(n)}_\UnitD ), \confmap^{-1} \circ P_\epsilon (\tilde{\gamma}_\UnitD ) ) \leq D_\epsilon d( \tilde{\gamma}^{(n)}_\UnitD, \tilde{\gamma}_\UnitD ).
\end{align*}
Using the convergence $\tilde{\gamma}^{(n)}_\UnitD \to \tilde{\gamma}_\UnitD$ in probability and taking $n > n_0^{(4)}(\epsilon)$ so that $\PR [ d( \tilde{\gamma}^{(n)}_\UnitD, \tilde{\gamma}_\UnitD ) \geq 1/(D_\epsilon 2^{m+2}) ] \leq 1/2^{m+2} $, we obtain
\begin{align}
\label{eq:term 3}
\PR [d( \confmap^{-1} \circ P_\epsilon (\tilde{\gamma}^{(n)}_\UnitD ), \confmap^{-1} \circ P_\epsilon (\tilde{\gamma}_\UnitD ) ) \geq 1/2^{m+2}] \leq 1/ 2^{m+2} \qquad \text{for all } n > n_0^{(3)} (\epsilon).
\end{align}

Note that the three terms above could be treated with simple and general arguments, not even requiring $\epsilon$ to be small. The core of the proof is thus the second term, which we formulate as a separate result and prove in Section~\ref{sec: pf of key lemma}.
\begin{keylem}
\label{lem: key lemma 1st statement}
In the setup of Theorem~\ref{thm: extension of Kemppainen's precompactness}, for any $\tilde{\eps} > 0$, there exists $\eps_0( \tilde{\epsilon}) > 0$ such that if $\eps < \eps_0$ and $n > n_0(\eps)$, then 
\begin{align*}
\PR^{(n)} [ d ( \confmap_n^{-1}(x) , \confmap_n^{-1} \circ P_{\eps} ( x  ) ) > \tilde{\eps} \text{ for some }x \in \tilde{\gamma}^{(n)}_\UnitD  ] \le \tilde{\eps}.
\end{align*}
\end{keylem}
In particular, setting $\tilde{\epsilon} = 1/2^{m+2}$, we have
\begin{align}
\label{eq:term 4}
\PR & [d(\gamma^{(n)}, \confmap_n^{-1} \circ P_\epsilon (\tilde{\gamma}^{(n)}_\UnitD ) ) \geq 1/2^{m+2}] \leq 1/ 2^{m+2} \\ 
\nonumber
& \text{by first taking $\epsilon < \epsilon_0 (m)$ small enough and then $n > n_0^{(2)} (\epsilon)$ large enough} .
\end{align}
Combining~\eqref{eq:term 1}--\eqref{eq:term 4} we observe that taking first $\epsilon_{j_m} < \epsilon_0 (m)$ and then $n > \max_{1 \leq k \leq 4} n_0^{(k)} (\epsilon_{j_m})$,~\eqref{eq:BC trick} is guaranteed.
\end{proof}

\begin{proof}[Proof of the conformal property of limits in Theorem~\ref{thm: extension of Kemppainen's precompactness}]
As noted above, it is equivalent to show that, almost surely, $\gamma = \confmap^{-1} ( \tilde{\gamma}_\UnitD)$.
Proposition~\ref{prop: weak limit commutation property} directly finishes the proof if the limiting domain $\domain$ is regular enough so that the conformal map $\confmapD$ extends continuously to the closed unit disc $\overline{\UnitD}$; this continuity implies that a curve $\lim_{\epsilon \downarrow 0} \confmap^{-1} \circ P_\eps ( \tilde{\gamma}_\UnitD )$ exists.
(Note that this argument did not impose any requirements for the approximating domains $\domain_n$.)
For more irregular limiting domains, the proof is concluded from Proposition~\ref{prop: weak limit commutation property} by Lemma~\ref{lem:det curve conv} below, which is a statement about deterministic curves.
\end{proof}

In the following, we mean by \textit{non-self-crossing curves} the $X(\C)$-closure of simple curves. Note that, by the Portmanteau theorem, any weak limit of probability measures supported on non-self-crossing curves curves is also supported on non-self-crossing curves. In particular, this is the case for any weak limit $\gamma_\UnitD$ in our setup, since $\gamma^{(n)}_\UnitD$ are Loewner transformable and hence non-self-crossing.

\begin{lem}
\label{lem:det curve conv}
Let $\domain$ be a bounded simply-connected domain, $\confmap$ its Riemann map normalized at $u \in \domain$, and $a, b$ two distinct prime ends with radial limits.
Let $\tilde{\eta}_\UnitD$ be a non-self-crossing 
curve in $\overline{ \UnitD}$ with $\tilde{\eta}_\UnitD (0) = \confmap( a)$ and  $\tilde{\eta}_\UnitD (1) = \confmap(b) $ (in the sense of the prime end bijection). Suppose that for some sequence $\epsilon_m \downarrow 0$, we have $\confmap^{-1} \circ P_{\epsilon_m} ( \tilde{\eta}_\UnitD ) \to \eta$ in $X(\C)$ as $m \to \infty$. Then, $\confmap^{-1}$, as extended by radial limits, is defined on all points of $\tilde{\eta}_\UnitD$ and 
\begin{align*}
\confmap^{-1} \circ P_\eps ( \tilde{\eta}_\UnitD (t) ) \stackrel{\eps \downarrow 0}{ \longrightarrow }  \confmap^{-1}  ( \tilde{\eta}_\UnitD (t) )
\quad
\text{uniformly over $t \in [0,1]$.} 
\end{align*}
 In particular, $ \confmap^{-1}  ( \tilde{\eta}_\UnitD ) $ is one parametrized representative of $\eta \in X(\C)$ and the uniform converge $\confmap^{-1} \circ P_{\epsilon} ( \tilde{\eta}_\UnitD ) \to \confmap^{-1}  ( \tilde{\eta}_\UnitD )$ takes place as $\epsilon \downarrow 0$ (not only along a special sequence) and in this particular parametrization.
\end{lem}

\begin{proof}
It suffices to show that the complex numbers $\confmap^{-1} \circ P_\eps ( \tilde{\eta}_\UnitD (t) )$ converge uniformly over $t \in [0,1]$ to some limits as $\eps \downarrow 0$.
Assume for a contradiction that this is not the case, and denote
\begin{align*}
\varrho (r, t) = \{  \confmap^{-1}  ( P_\eps ( \tilde{ \eta }_\UnitD (t ) ) ) , \eps \in (0, r) \} \subset \C.
\end{align*}
The counter assumption means that for some $\ell > 0$ there exists a sequence $r_J \downarrow 0$ and times $t_J \in [0, 1]$ such that 
\begin{align*}
\diam ( \varrho (r_J, t_J) ) > \ell.
\end{align*}
Denote $\theta(t) = \arg ( \tilde{ \eta }_{\UnitD } (t) ) \in [0, 2 \pi) $; by compactness we may assume that $t_J$ are chosen so that $e^{i\theta(t_J)} \to e^{i \theta}$, and that the sequence $\theta(t_J)$ is monotonous. 

Fix $J$ and suppress the indexing for a moment, $(r,t):=(r_J, t_J)$.
Let $x$ and $y$ be two points on $\varrho (r, t)$ such that $d(x, y ) > \ell / 2$, $y$ corresponding to a smaller value of $\eps= \eps^{(y)}$ in $y= P_{\eps^{(y)}} ( \tilde{ \eta }_\UnitD (t ) )$.
By Proposition~\ref{prop: conformal rays lie in fjords} from Section~\ref{sec: pf of key lemma}, there exists a quantity $\rho=\rho(r)$ (independent of $t$), $\rho \to 0$ as $r \to 0$, so that the following holds for all $r$ small enough: the innermost connected component $S^{(x)}$ of the circle $S ( x, \rho)$ that separates $x$ from the normalization point $u$, actually separates from $u$ the remainder of the whole ``conformal ray'' (i.e., a $\confmap^{-1}$ image of a radial line segment) from $x$ onwards; and the analogous statement holds with $x$ replaced by $y$ for the same $\rho$. Assume that $r$ is small enough so that $\rho < \ell /8$; hence $S^{(x)}$ disconnects $S^{(y)}$ from $u$. Changing the choice of $\rho=\rho(r)$ infinitesimally if needed, we will also later assume that neither $\confmap (S^{(x)})$ nor $\confmap (S^{(y)})$ has $e^{i \theta}$ as an end point on $\bdry \UnitD$.

We next claim that for $r<r_0$ small enough, where $r_0$ only depends on the domain $(\domain; a, b)$, the limit curve $\eta$ must contain a subcurve $\eta([t_1, t_2])$ between $\overline{S^{(x)}}$ and $\overline{S^{(y)}}$. (Note that such a subcurve necessarily has a diameter $\ge \ell / 4$.)
To work towards this observation, note that $S^{(x)}$ and $S^{(y)}$ divide $\domain$ into three components: the component $C_u$ of $u$ in $\domain \setminus S^{(x)}$, the boundary quad $Q$ between  $S^{(x)}$ and $S^{(y)}$, and the component $C_{y}$ of $\domain \setminus S^{(y)}$ containing $y$ (and not containing $u$ nor $S^{(x)}$).
By the construction of $S^{(y)}$, the points $P_{\eps} ( \tilde{ \eta }_\UnitD (t ) )$ with $ \eps \leq \eps^{(y)}$ (in particular along the special sequence $\epsilon_m \downarrow 0$) all belong to $C_y$. On the other hand, for $r$ and hence $\rho(r)$ small enough, at least one out of $a$ and $b$ must lie on the boundary $C_u$  (otherwise, taking $\rho \to 0$, one boundary arc between $a$ and $b$ would have an arbitrarily small harmonic measure as seen from $u$). Hence for all $m$ large enough, $\confmap^{-1} \circ P_{\epsilon_m} ( \tilde{\eta}_\UnitD )$ is a curve in $\domain$ that visits both $C_y$ and $C_u$, in particular crossing between $\overline{S^{(x)}}$ and $\overline{S^{(y)}}$. The existence of such a crossing carries over to the limit curve $\eta$.

We now go back to the sequence $(r_J, t_J)$, constructing a certain subsequence $J_k$. Let $S^{(x)}_k$ and $S^{(y)}_k$ be the circle arcs above for $(r,t):=(r_{J_k}, t_{J_k})$.\footnote{We should denote here $x=x_k$ and $y=y_k$; the exact choice of these points is however unimportant in what follows, apart from $x$ lying closer to $u$ on the same ``conformal ray''.}
We define $J_{k}$ inductively so that the cross cut $\confmap (S^{(x)}_{k})$ separates $\confmap (S^{(y)}_{k})$ in $\UnitD$ from  all the previous ones, i.e., $\confmap (S^{(x)}_{1}), \ldots, \confmap (S^{(x)}_{k-1})$ and $\confmap (S^{(y)}_{1}), \ldots, \confmap (S^{(y)}_{k-1})$. This is possible since the previous cross cuts avoid $e^{i \theta}$, while an easy harmonic measure argument shows that for $J_k$ large enough, entire $\confmap (S^{(x)}_{k})$ can be made arbitrarily close to $e^{i \theta}$.
As observed above, the limit curve $\eta$ must contain a subcurve $\eta([t^{(k)}_1, t^{(k)}_2])$ from $\overline{S^{(x)}_{k}}$ to $\overline{S^{(y)}_{k}}$, hence with diameter $\ge \ell / 4$. We claim that the time intervals $[t^{(k)}_1, t^{(k)}_2]$ can be chosen disjoint. Proving this claim true contradicts the continuity of $\eta$ and finishes the entire proof.

To construct the disjoint time intervals, we consider the following three cases, illustrated in \ref{fig: three cases}:
\begin{itemize}
\item infinitely many of the boundary quads $Q_n$ in $\domain$ determined by the two cross cuts $S^{(x)}_{n}$ and $S^{(y)}_{n}$ contain either $a$ or $b$;
\item infinitely many of the boundary quads contain neither $a$ nor $b$ and disconnect $a$ from $b$; and
\item infinitely many of the boundary quads contain neither $a$ nor $b$ and do not disconnect $a$ from $b$.
\end{itemize}
At least one of these occurs. Let us denote the corresponding infinitely many quadrilaterals  by $Q_1, Q_2, \ldots$, suppressing a subsequence notation.
In the first case, assume that infinitely many quads contain $b$ (if not, reverse the curve $\eta$). Note that the quadrilaterals are nested by the monotonicity of $\theta(t_J)$, $Q_1 \supset Q_2 \supset \ldots$ and that the inner arcs $S^{(y)}_k$ ``move towards $b$'', see Figure~\ref{fig: three cases}. Now, an initial segment of $\eta$ up to the end of any crossing of $Q_k$ from $S^{(x)}_{k}$  to $S^{(y)}_{k}$ either contains no crossing the next quadrilateral $Q_{k+1}$ from $S^{(x)}_{k+1}$  to $S^{(y)}_{k+1}$, or the rest of $\eta$ is topologically forced to make such a crossing again (due to the non-self-crossing property). In any case, we can thus find disjoint crossings of the quads  in their order. The second case is trivial: the quadrilaterals are then by construction disjoint, and so are their crossings. In the third case, assume for definiteness that the boundary points $e^{i\theta(t_{J_k})}$ corresponding to the quadrilaterals move farther from $\confmap(a)$ and towards $\confmap(b)$ (if not, reverse $\eta$). The proof is then identical to the first case. This completes the proof.
\begin{figure}
\makebox[\textwidth][c]{
\includegraphics[ width=0.3\textwidth]{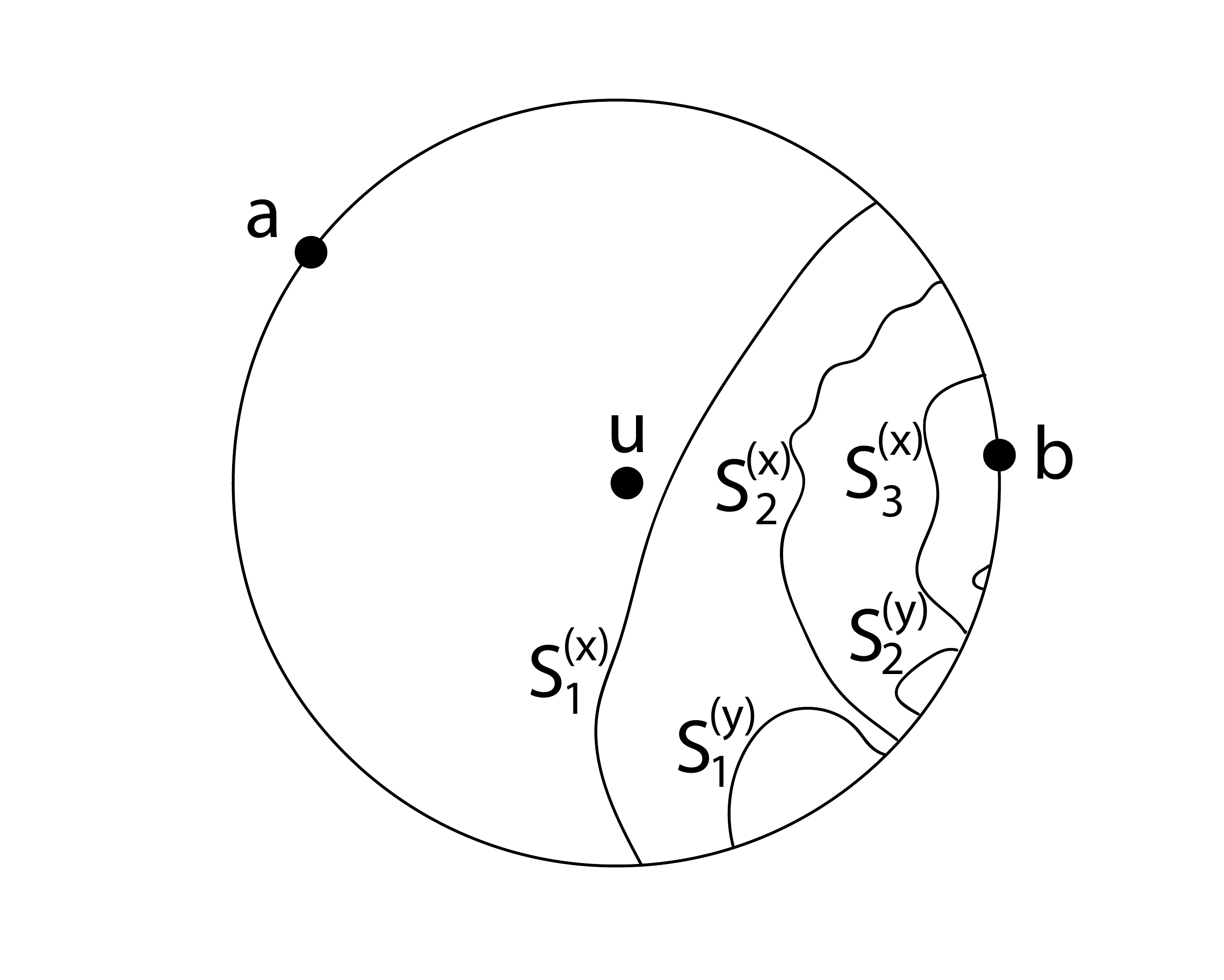}%
\includegraphics[ width=0.3\textwidth]{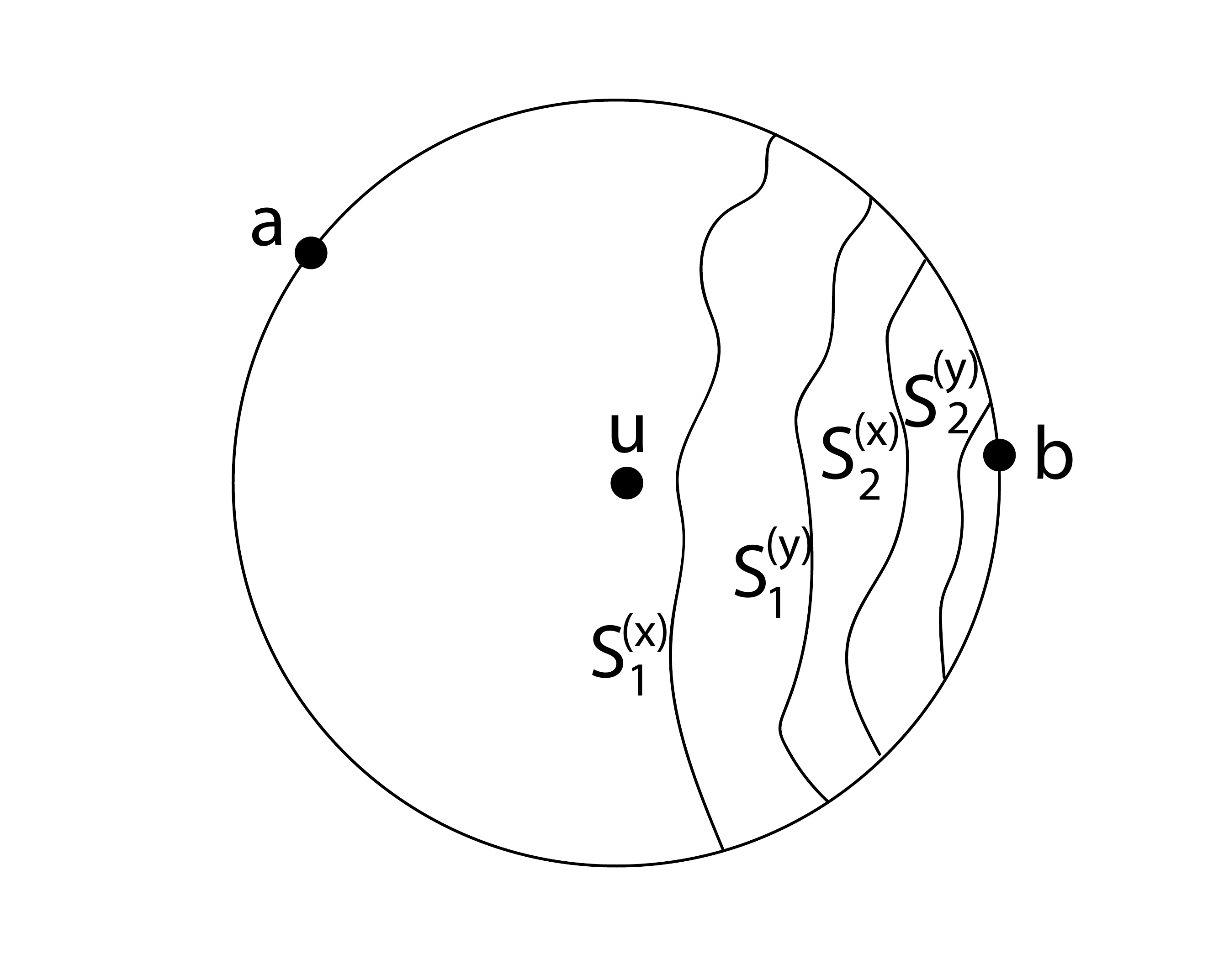}%
\includegraphics[ width=0.3\textwidth]{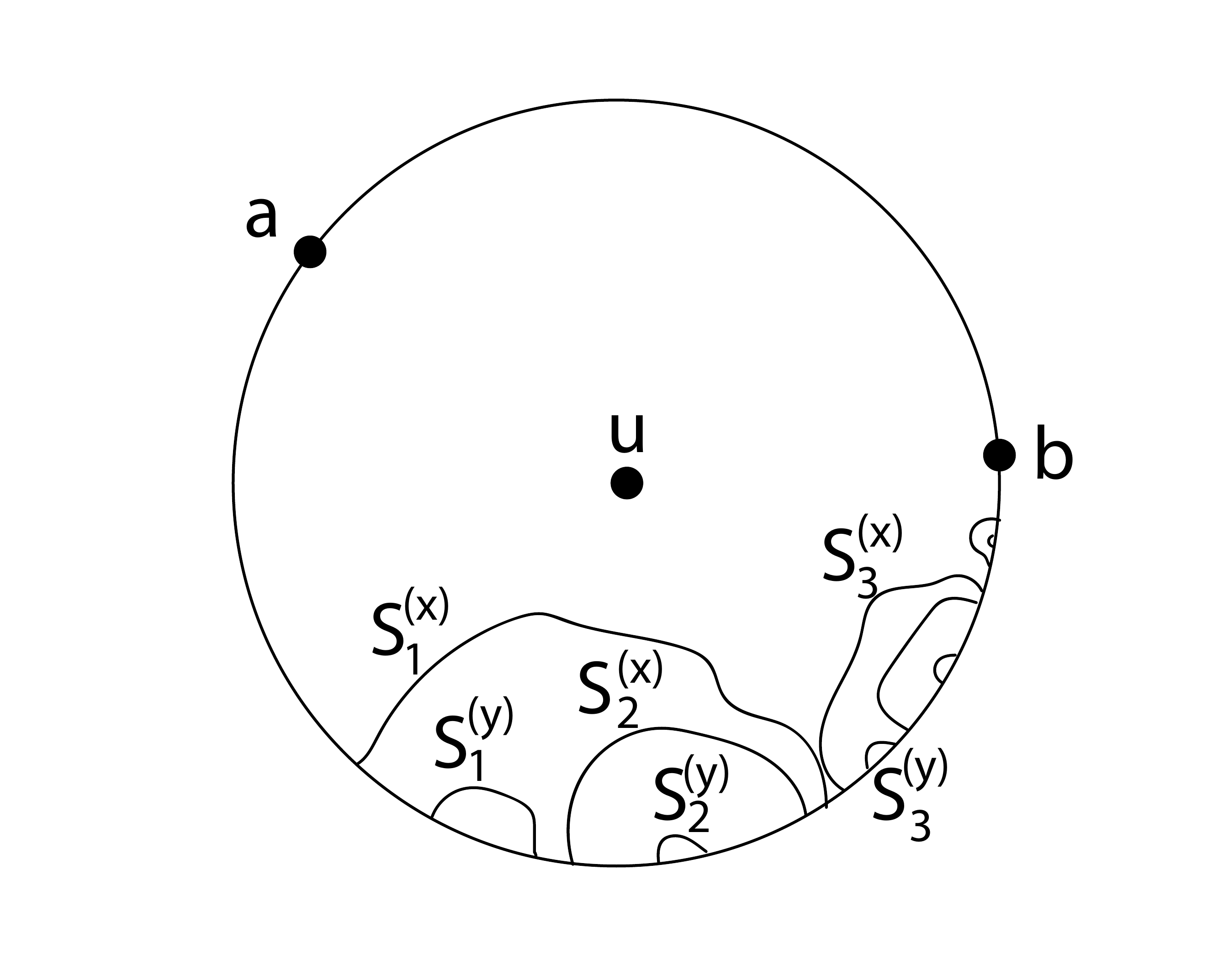}
}
\caption{\label{fig: three cases} Schematic illustrations of the prime ends $a$ and $b$, the point $u$, and the collection of cross cuts $S^{(x)}_n, S^{(y)}_n$, as mapped conformally to $\UnitD$ by $\confmap$ (the notations $\confmap$ are omitted in the figure), in the different cases itemized in the end of the proof of Proposition~\ref{lem:det curve conv}. }
\end{figure}
\end{proof}

\subsection{Parametrized curves}

For a pair $( \gamma_\UnitD, W) \in X(\overline{\UnitD}) \times \ctsfcns $ of Loewner transforms, we say that $\gamma_\UnitD$ is \textit{capacity-parametrized} if $\gamma_\UnitD (t) = \confmapHD (\gamma_{\bH} ( \frac{t}{1-t} ))$ for $t \in [0,1)$, where $\gamma_{\bH} $ is the map generating the hulls of $W$. In this subsection, we assume that $\gamma_\UnitD^{(n)}$ and $\gamma_\UnitD$ are capacity-parametrized and equip the space $\ctsfcns([0,1], \C)$ of continuous functions $[0,1] \to \C$ by the usual sup norm. One (simple extension of a) result in~\cite{KS} is that Theorem~\ref{thm: Kemppainen's precompactness} also holds in the (stronger) topology of $\ctsfcns([0,1], \C ) \times \ctsfcns$. 

\begin{thm}
\label{thm:parametrized}
Suppose that $\gamma_\UnitD^{(n)}$ and $\gamma_\UnitD$ are capacity-parametrized and $\gamma^{(n)}$ are parametrized so that $\gamma^{(n)}(\cdot)=\confmapBdry^{-1}_n (\gamma^{(n)}_\UnitD (\cdot))$. Then, the statement of Theorem~\ref{thm: extension of Kemppainen's precompactness} holds in the topology $\ctsfcns([0,1], \C ) \times \ctsfcns([0,1], \C ) $.
\end{thm}

\begin{proof}
By Theorem~\ref{thm: extension of Kemppainen's precompactness} and~\cite{KS}, $(\gamma_\UnitD^{(n)}, \gamma^{(n)})$ are tight in $\ctsfcns([0,1], \C ) \times X (\C ) $. Suppose hence that $(\gamma_\UnitD^{(n)}, \gamma^{(n)}) \to (\gamma_\UnitD, \gamma)$ weakly in $\ctsfcns([0,1], \C ) \times X (\C ) $. We will show that there exists a subsequence $(n_m)$ such that also $(\gamma_\UnitD^{(n_m)}, \gamma^{(n_m)}) \to (\gamma_\UnitD, \confmapBdry^{-1} (\gamma_\UnitD ))$ in $\ctsfcns([0,1], \C ) \times \ctsfcns([0,1], \C ) $; this readily implies the entire analogue of Theorem~\ref{thm: extension of Kemppainen's precompactness} in $\ctsfcns([0,1], \C ) \times \ctsfcns([0,1], \C ) $.

Similarly to the proof of Theorem~\ref{thm: extension of Kemppainen's precompactness}, embed everything in one probability space so that $(\gamma_\UnitD^{(n)}, \gamma^{(n)}) \to (\gamma_\UnitD, \gamma)$ in $\ctsfcns([0,1], \C ) \times X (\C ) $, almost surely. Now, we claim that, for all $m \in \N$, one can choose $n_m$ so that, in this coupling
\begin{align*}
\bbP [d (\gamma^{(n_m)}, &\confmap^{-1} (\tilde{\gamma_\UnitD})) \geq 1/2^m] \leq 1/2^m,
\end{align*}
where the distance is in $\ctsfcns([0,1], \C )$ and we remind the reader that $\tilde{\gamma}_{\UnitD} = \confmap \circ \confmapBdry^{-1} ( \gamma_\UnitD  ) $.
The desired weak convergence along subsequence $n_m$ then readily follows by the Borel--Cantelli lemma.

To prove the above claim, compute with distances in $\ctsfcns([0,1], \C )$:
\begin{align*}
d (\gamma^{(n)}, \confmap^{-1} (\tilde{\gamma_\UnitD}))
\leq &
d(\gamma^{(n)}, \confmap_n^{-1} \circ P_{\epsilon} (\tilde{\gamma_\UnitD}^{(n)}))
+ 
d( \confmap_n^{-1} \circ P_{\epsilon} (\tilde{\gamma_\UnitD}^{(n)}), \confmap^{-1} \circ P_{\epsilon} (\tilde{\gamma_\UnitD}^{(n)}))
\\ & +
d(\confmap^{-1} \circ P_{\epsilon} (\tilde{\gamma_\UnitD}^{(n)}), 
\confmap^{-1} \circ P_{\epsilon} (\tilde{\gamma_\UnitD}))
+ 
d(\confmap^{-1} \circ P_{\epsilon} (\tilde{\gamma_\UnitD}), \confmap^{-1} (\tilde{\gamma_\UnitD})).
\end{align*}
The rest of the proof is similar to the proof Proposition~\ref{prop: weak limit commutation property} and we only outline the idea: the three first terms on the right-hand side can be handled identically to that proof, choosing for each $m$ first $\epsilon < \epsilon_0 (m)$ and then $n > n_0 (\epsilon)$. For the fourth one, the almost sure convergence $(\tilde{\gamma}_\UnitD^{(n)}, \gamma^{(n)}) \to (\tilde{\gamma}_\UnitD, \gamma)$ in $\ctsfcns([0,1], \C ) \times X (\C ) $ implies the analogous weak convergence in $X(\overline{\UnitD}) \times X (\C )$; then by Proposition~\ref{prop: weak limit commutation property} and Lemma~\ref{lem:det curve conv} (which directly addresses parametrized curves), $\confmap^{-1} \circ P_{\epsilon} (\tilde{\gamma_\UnitD}) \stackrel{ \epsilon \downarrow 0}{\longrightarrow} \confmap^{-1} (\tilde{\gamma_\UnitD})$ in $\ctsfcns([0,1], \C )$, almost surely and thus also in probability.
%Handling the three first terms identically to the proof of Proposition~\ref{prop: weak limit commutation property}, we then choose for each $m$ first $\epsilon < \epsilon_0 (m)$ and then $n > n_0 (\epsilon)$, and the proof is concluded.
\end{proof}

\section{An application to SLE}
\label{sec: applications}

In this section, we demonstrate the use of Theorem~\ref{thm: extension of Kemppainen's precompactness} by proving that the chordal $\SLE(\kappa)$ (as a random curve in $X(\C)$) is stable both in the domain and in the parameter $\kappa \in [0,8)$. This can be seen as an analogue of the SLE application in~\cite{KS}. Applying Theorem~\ref{thm:parametrized} instead, the analogues of the results in this section could be given for capacity-parametrized curves.

\subsection{SLE in the topology of curves}
\label{subsec:SLE as rand curve}

Recall from Section~\ref{subsubsec:rand growth proc} the definition of SLE as a random Loewner growth processes.
By the celebrated results of~\cite{RS-basic_properties_of_SLE}, the hulls of an $\SLE (\kappa)$ growth process with $\kappa \in [0,8)$ are almost surely generated by a continuous map $\gamma_\bbH$ and $\gamma_{\bH} (t) \stackrel{t \to \infty}{\longrightarrow} \infty$. For such $\gamma_{\bH}$ there exists also a corresponding curve $\gamma_\UnitD $ in $\overline{\UnitD}$ (see Section~\ref{subsubsec:rand growth proc}).

\begin{prop}
\label{prop: SLE in D is measurable}
Fix $\kappa \in [0,8)$. The SLE curve $\gamma_\UnitD $ defined above is almost surely equal to an $X(\overline{ \UnitD })$-valued random variable, measurable in the sigma algebra of the Brownian motion that defines the SLE.
\end{prop}

\begin{proof}[Proof sketch.]
Note first that condition (C) can be formulated only in terms of hulls and stopping times of the Loewner process. It thus makes sense to claim that $\gamma_\UnitD$ satisfies condition (C), even if do not know if it exists as a $X(\overline{ \UnitD })$-valued random variable. Condition (C) is then verified in~\cite[Theorem~1.10]{KS}.\footnote{Prior to~\cite{KS}, SLE semicircle intersection probabilities were studied in~\cite{AK-SLE_intersection}. Condition (C) alternatively follows from this result and the relation of annulus and quad crossings, see~\cite[Proof of ``G2$\Rightarrow$C2'']{KS}.}
Finally, arguing identically to the proof of~\cite[Theorem~1.7]{KS}, condition (C) implies that for any $\eps > 0$, there is a set $K_\eps$ in the space of curves with a Loewner transform, compact both in the topology of the driving functions $W \in \ctsfcns $ and the curves $\gamma_\UnitD \in X(\C)$, and carrying probability mass $\PR [ K_\eps ] \ge 1 - \eps$, and on which $\gamma_\UnitD \in X(\C)$ is a continuous function of $W \in \ctsfcns $. The claim follows.
\end{proof}

The next result yields an alternative proof of~\cite[Theorem~1.1]{GRS-SLE_in_sc_domains}, in addition addressing the behaviour of the curve at the end points and its measurability as a random variable.

\begin{prop}
\label{prop: SLE in domain is measurable}
Let $\domain$ be a bounded simply-connected domain, $a, b$ its prime ends with radial limits, $\confmapBdry$ a conformal map $(\domain; a, b) \to (\UnitD; -1, 1)$, and $\gamma_\UnitD $ an $\SLE(\kappa)$ curve on $(\UnitD; -1, 1)$ with $\kappa \in [0,8)$. There exists an $X(\C)$-valued random variable $\gamma$, measurable with respect to the sigma algebra of $\gamma_\UnitD $ (and hence also that of the underlying Brownian motion) such that, almost surely, $\gamma = \confmapBdry^{-1} (\gamma_\UnitD  )$, where $ \confmapBdry^{-1}$ is extended by radial limits.
\end{prop}

\begin{proof}
Set $u = \confmapBdry^{-1} (0) \in \domain$; rotating $\domain$, we may assume that $\confmapBdry$ is the Riemann map normalized at $u$. Let $\domain_n$ be the simply-connected domain bounded by the simple loop on $\frac{1}{n} \Z^2$ that encloses $u$ and a maximal number of squares of $\frac{1}{n} \Z^2$, let $\confmapBdry_n$ be their Riemann maps normalized at $u$, and $a_n = \confmapBdry_n^{-1}(-1)$, $b_n = \confmapBdry_n^{-1}(1)$; by polygonality, $\confmapBdry_n^{-1}$ extend continuously to $\overline{\UnitD}$ and the curves $\confmapBdry_n^{-1} (\gamma_\UnitD ) = \gamma^{(n)} \in X(\C)$ are thus measurable random variables in the sigma algebra of $\gamma_\UnitD $. They
satisfy Condition (C) by the argument in the proof of Proposition~\ref{prop: SLE in D is measurable}, and the remaining assumptions of Theorem~\ref{thm: extension of Kemppainen's precompactness} are readily verified. Hence, $(\gamma_\UnitD, \gamma^{(n)})$ converge weakly to $(\gamma_\UnitD, \gamma)$, where $\gamma = \confmapBdry^{-1} (\gamma_\UnitD)$, and $\gamma $ is  measurable in the sigma algebra of $\gamma_{\UnitD}$ by Proposition~\ref{prop: weak limit commutation property}. This concludes the proof.
\end{proof}

\subsection{Stability of SLE}
\label{subsec: SLE stability}

We finish by proving the stability of the chordal SLE 
in both the domain and the parameter $\kappa$. Stability in $\kappa$ has been addressed in~\cite{JRW-SLE_stability, KS}, and domain stability (for $\kappa \le 4$) in~\cite{LK-configurational_measure}. Note that in order to apply the main results of this paper, derived for random triplets $(W^{(n)}, \gamma_\UnitD^{(n)}, \gamma^{(n)})$ such that $W^{(n)}$ and $ \gamma_\UnitD^{(n)}$ are Loewner transforms of each other and $\gamma^{(n)}= \confmapBdry_n^{-1} (\gamma^{(n)}_\UnitD ) $, the results of the previous subsection were necessary.

\begin{prop}
\label{prop: SLE is stable}
Assume that $\kappa_n \to \kappa < 8$ and $(\domain_n; a_n, b_n) \Cara (\domain; a, b)$ satisfy the assumptions of Theorem~\ref{thm: extension of Kemppainen's precompactness}. Then, the $\SLE (\kappa_n)$-curves $\gamma^{(n)} $ in $(\domain_n; a_n, b_n)$ converge weakly in $ X(\C)$ to the $SLE(\kappa)$ curve $\gamma$ in $(\domain; a, b)$.
\end{prop}

\begin{proof}
Embed all the SLEs in the same probability space by sampling from the same Brownian motion, $W^{(n)}_t = \sqrt{\kappa_n} B_t$. Hence, $W^{(n)}_t \to W_t = \sqrt{\kappa} B_t$ almost surely and weakly in $\ctsfcns$. Condition (C) is again satisfied by the argument in Proposition~\ref{prop: SLE in D is measurable}; the claim then follows from Theorems~\ref{thm: Kemppainen's precompactness}--\ref{thm: extension of Kemppainen's precompactness}.
\end{proof}

\section{Proof of the Key lemma~\ref{lem: key lemma 1st statement} in Theorem~\ref{thm: extension of Kemppainen's precompactness}}
\label{sec: pf of key lemma}

In this section, we prove the Key lemma~\ref{lem: key lemma 1st statement}, which also constitutes the bulk of the proofs of Proposition~\ref{prop: weak limit commutation property} and Theorem~\ref{thm: extension of Kemppainen's precompactness}. For convenience, we take an equivalent formulation where the probability bound and the distance bound are given the distinct notations $\tilde{\eps}$ and $\ell$, respectively.

\begin{keylem}
\label{lem: key lemma}
In the setup and notation of Proposition~\ref{prop: weak limit commutation property}, for any $\ell > 0$ and any $\tilde{\eps} > 0$, there exists $\eps_0 (\ell, \tilde{\eps}) > 0$ such that if $\eps < \eps_0$ and $n > n_0(\eps)$, then 
\begin{align*}
\PR^{(n)} [ d ( \confmap_n^{-1}(x) , \confmap_n^{-1} \circ P_{\eps} ( x  ) ) > \ell \text{ for some }x \in \tilde{\gamma}^{(n)}_\UnitD ] \le \tilde{\eps}.
\end{align*}
\end{keylem}

The rest of this section constitutes the proof. Section~\ref{subsec:fjords} defines fjords of a domain and recalls a probabilistic result from~\cite{KS}, showing that ``deep tours into narrow fjords'' by random curves are unlikely, given the equivalent conditions (C) and (G). Section~\ref{subsec:rad proj}, in turn, is purely analytic and argues that the radial projection $\confmap_n^{-1} \circ P_{\eps} (x)$ does not differ much from $\confmap_n^{-1} (x)$ except if the latter lies deep in a narrow fjord of $\domain_n$. The proof is concluded in Section~\ref{subsec:concluding key lem} by putting these pieces together and handling the curve end points. 

\subsection{Fjords}
\label{subsec:fjords}

Let $\domainn$ be a simply-connected domain, $u \in \domainn$ and $a_\n$ and $b_\n$ prime ends of $\domainn$ (here and in continuation, we will use subscript $*$ instead of $n$ to indicate that it is fixed).
Let $S$ be a cross cut of $\domainn$, with $\diam(S) \le \delta$ and $u \not \in S$. Then, we say that the union $F$ of $S$ and the connected component of $\domainn \setminus S$ that does not contain $u$ is a \textit{$\delta$-fjord with respect to $u$}. Similarly, if some neighbourhoods of $a_\n$ and $b_\n$, respectively, are both contained in one and same connected component of $\domainn \setminus S$, then  the union $F$ of $S$ and the connected component  of $\domainn \setminus S$ not adjacent to $a_\n$ and $b_\n$ forms a \textit{$\delta$-fjord with respect to $(a_\n, b_\n)$}.
We say that a point $z \in F$ lies \textit{$\ell$-deep} in $F$ if the interior distance in $\domainn$ from $z$ to $S$ is at least $\ell$, $d_{\domainn} (z, S) \ge \ell$. 

\begin{prop}
\label{lemma: Kemppainen's fjord lemma}
\emph{(Special case of \cite[Lemma~3.14]{KS})}
In the setup of Section~\ref{subsec: setup}, assume that the domains $\domain_n$ are uniformly bounded in area, Area$(\domain_n ) \leq M$ for all $n$, and that the measures $\PR^{(n)}$ satisfy the equivalent conditions (C) and (G). Then, for any $\ell > 0$ and any $\tilde{\eps} > 0$, there exists $\delta = \delta(M, \ell, \tilde{\eps})> 0$ (independent of $n$) such that if $(F^{(n)}_i)_{i=1}^{m^{(n)}}$ is any finite collection of interior-disjoint $\delta$-fjords with respect to $(a_n, b_n)$ in $\domain_n$, then
\begin{align}
\label{eq:fjord tour proba bound}
\PR^{(n)} [ \gamma^{(n)} \text{ visits $\ell$-deep in some of the fjords $(F^{(n)}_i)_{i=1}^{m^{(n)}}$}] \le \tilde{\eps}.
\end{align}
\end{prop}

\subsection{Boundary behaviour of radial projections}
\label{subsec:rad proj}

This subsection constitutes the complex analysis part of the proof of Lemma~\ref{lem: key lemma}. We start with some notations and definitions. Given a conformal map $\confmap_\n: \domainn \to \UnitD$, we define \textit{conformal ray segments}
\begin{align*}
\rho_{ \theta, p, q }^{(\n)} &= \confmapD_\n ( \{ z \in \UnitD:  \; z = t e^{i \theta} \text{ for some } t \in [p, q] \} ), \quad \text{for } 0 \le p < q < 1  \text{ and } \theta \in [0, 2 \pi);
\end{align*}
we will also allow $q=1$ by setting $t \in [p,q)$ above. The entire \textit{conformal ray} is denoted by $\rho_{ \theta, 0,1 }^{(\n)} = \rho_{ \theta }^{(\n)}$ for short. Recall also the concept of \textit{innermost disconnecting arc} from Lemma~\ref{lem: disconnecting arcs} in Appendix~\ref{app:roskis}. 

%We also need to introduce the concept a connection $\delta$-inside a domain:
Let $K \subset \domainn$ be closed and $z \in \domainn \setminus K$. We say that a curve $\curve: [0, 1] \to \domainn$ \emph{connects $z$ to $K$ $\delta$-inside} $\domainn$ if $\curve(0) = z$, $\curve (1) \in K$, and $\curve([0, 1)) \subset \domainn \setminus K$, and for the open $\delta$-thickening of curve $\curve$
\begin{align*}
G_\delta = \{  w \in \C: d(w, \curve(t)) < \delta \text{ for some } t \in [0, 1] \},
\end{align*}
the connected component of $G_\delta \setminus K$ that contains $\curve([0, 1))$ is entirely contained in $ \domainn$; see Figure~\ref{fig: connection delta-inside}.

\begin{figure}
\includegraphics[ width=0.28\textwidth]{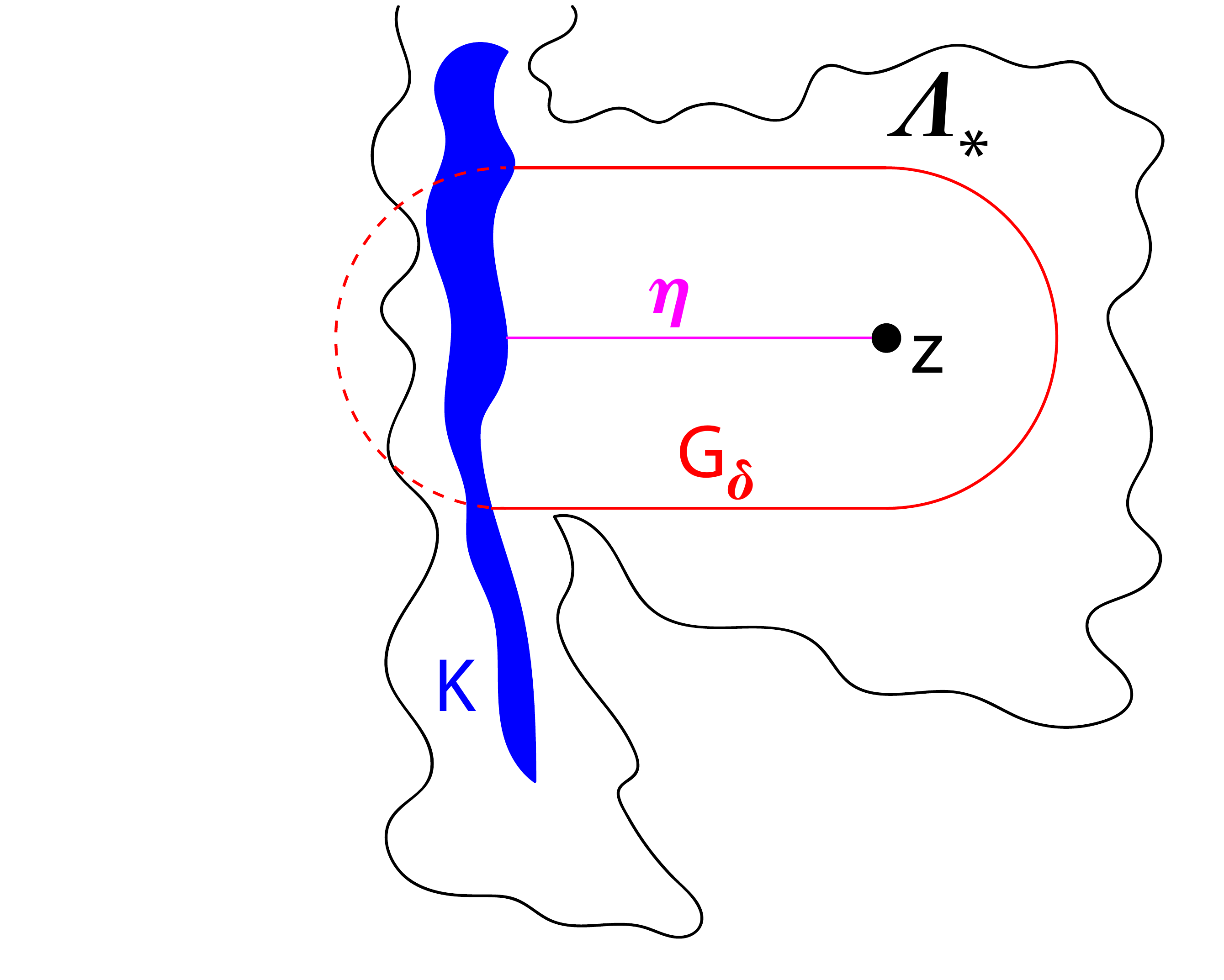}
\caption{\label{fig: connection delta-inside} A curve $\curve$ connecting $z$ to $K $  $\delta$-inside $\domainn$.}
\end{figure}

Finally, we say that \textit{$B(0, 1-\eps)$ is a $\delta$-approximation of $\domainn$ under $\confmapD_{\n}$} if
\begin{itemize}
\item[i)] for all $\zeta \in \confmapD_{\n} ( S (0 , 1 - \eps ))$, we have $d(\zeta, \bdry \domainn) <\delta $; and
\item[ii)] if $z \in \domainn \setminus \confmapD_{\n} ( \overline{B (0 , 1 - \eps ) } )$ satisfies $d(z, \bdry \domainn) \ge \delta $, then there exists no curve $\curve$ that connects $z$ to $\confmapD_{\n} ( \overline{B (0 , 1 - \eps ) } )$ $\delta$-inside $\domainn$.
\end{itemize}
The idea of this definition it can be guaranteed for the tail of a Carath\'{e}odory converging sequence \textit{with the same $\delta$ and $\eps$}.

\begin{lem}
\label{lem: conformal balls get close to boundary}
Assume that $\domain_n \Cara \domain$, and that $\domain$ is bounded. For any $\delta > 0$, there exists $ \eps_0 (\delta)$ such that if $\eps < \eps_0$ is first taken small enough and after that $  n_0 (\delta, \eps)$ large enough, then $B(0, 1-\eps)$ are $\delta$-approximations of $\domain_n$ under $\confmapD_{n}$, for all $ n \ge n_0 (\delta, \eps)$.
\end{lem}

The proof of the above lemma is postponed to the Appendix.
We are now at a position to state the main result of this subsection (see Figure~\ref{fig: conf rays lie in fjords} for an illustration).

\begin{figure}
\includegraphics[ width=0.3\textwidth]{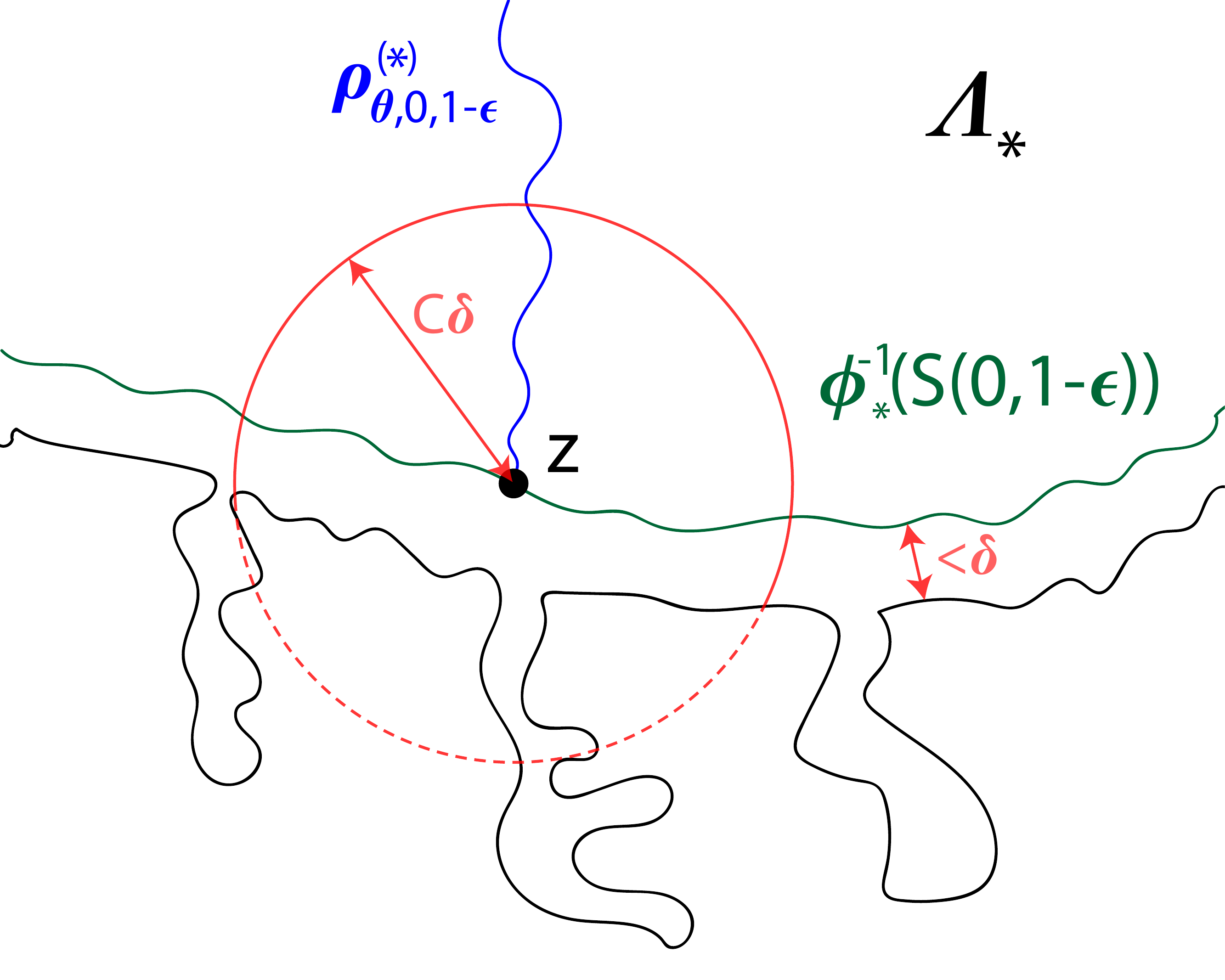}%
\caption{\label{fig: conf rays lie in fjords} A schematic illustration of Proposition~\ref{prop: conformal rays lie in fjords}. The remainder of the conformal ray $\rho^{(\n)}_\theta$ from $z$ onwards lies inside the fjord determined by the cross cut given by the part of the circle drawn in solid.}
\end{figure}

\begin{prop}
\label{prop: conformal rays lie in fjords}
There exist absolute constants $C, D>0 $ such that the following holds.
Let $\domainn$ be a simply-connected planar domain and $\confmap_\n : \domainn \to \UnitD$ its Riemann map normalized at $u \in \domainn$. Let $\delta \in (0, \tfrac{D}{C} d(u,\bdry \domainn ) )$ and suppose that $B(0, 1-\eps)$ is a $\delta$-approximation of $\domainn$ under $\confmapD_{\n}$. Then, for all  $z$ of the form $z = \confmapD_{\n} ( (1 - \eps) e^{i \theta} )$, the innermost arc of $S (z, C \delta)$ disconnecting $z$ from $u$ disconnects from $u$ the entire conformal ray segment $\rho^{(\n)}_{\theta, 1- \eps, 1}$ from $z$ onwards.
\end{prop}

The rest of this subsection complements this proposition:
some necessary results about harmonic measures are given in Section~\ref{subsubsec:Beurling}, and the proof of Proposition~\ref{prop: conformal rays lie in fjords} in Section~\ref{subsubsec:proof of rad prof prop}.

\subsubsection{\textbf{Harmonic measures}}
\label{subsubsec:Beurling}

Recall that the \textit{harmonic measure} $\HarmMeas(z, \domain, E)$ of a boundary set $E \subset \bdry \domain$ in a domain $\domain$ as seen from $z \in \domain$ is the probability that a Brownian motion launched from $z$ first hits $\bdry \domain$ on $E$. In all the cases that we consider, $\HarmMeas(z, \domain, E)$ coincides with the unique harmonic function in $z \in \domain$ that takes boundary values $1$ on $E$ and $0$  on $\bdry \domain\setminus \overline{E}$. An important property is that the harmonic measure is conformally invariant, in the sense that for a conformal map $\confmap$ defined on $\domain$, then $\HarmMeas(\confmap(z), \confmap(\domain), \confmap(E) ) = \HarmMeas(z, \domain, E)$. 

\begin{lem}
\emph{[The (weak) Beurling estimate]} There exist absolute constants $A, \alpha > 0$ such that the following holds. Let $z \in \C$, $R >0$ and denote $B=B(z, R)$; let $K \subset \C$ be a connected closed set that intersects $\bdry B$ and denote $r = d (z, K)$; then
\begin{align*}
\HarmMeas(z, B \setminus K, \bdry B) \le A \left( r/R \right)^\alpha.
\end{align*}
\end{lem}

See, e.g.,~\cite[Proposition~2.11]{CS-discrete_complex_analysis_on_isoradial} for a probabilistic proof;
an analytic proof follows from the Beurling projection theorem (\cite[Theorem~3.6]{Ahlfors} or~\cite[Theorem~1]{Oksendal}) and a sequence of explicit conformal maps as in~\cite[Section~4]{Benes-note}. The latter also reveals that $A=\frac{4}{\pi}$ and $\alpha=1/2$ yield a tight bound (in the sense of relative error when $K$ is a radial line segment).

We define the harmonic function $h: \UnitD \setminus [p, q ] \to \R$ by
\begin{align}
\label{eq:def of hm}
h (z) = \HarmMeas(z, \UnitD \setminus [p, q ]  , [p, q ] ).
\end{align}
This function, together with the conformal invariance of harmonic measures and the Beurling estimate, will play a key role in restricting the geometric behaviour of conformal rays, due to the following maximization property, whose proof we postpone to Appendix~\ref{app:roskis}.

\begin{lem}
\label{lem: maximization property of h_m}
For any $p, q, \epsilon \in [0, 1]$ with $1-\eps \le p < q \le 1$, the function $h$ defined in~\eqref{eq:def of hm} attains its unique maximum in $\overline{B (0, 1-\epsilon)}$ at the point $(1 - \eps)$.
\end{lem}

\subsubsection{\textbf{Proof of Proposition~\ref{prop: conformal rays lie in fjords}}}
\label{subsubsec:proof of rad prof prop}

Assume for a contradiction that for some conformal ray $\rho^{(\n)}_{\theta, 1- \eps, 1}$ in some domain $\domainn$, the innermost disconnecting component of $S (z, C \delta)$ does not separate $\rho^{(\n)}_{\theta, 1- \eps, 1}$ from $u$. We will show that there exist absolute constants $C_0, D_0 > 0$ such that if $C > C_0$ and $D<D_0$, this assumption leads to a contradiction. 

Define the following notations (see Figure~\ref{fig: cross cuts and conf ray}):
let $S_2$ and $S_0$ denote the innermost disconnecting components of $S (z, C \delta/3)$ and $S (z, C \delta)$ in $\domainn$, respectively (they exist if $C > 3$ and $ D $ is small enough). Denote by $Q$ the unique connected component of $\domainn \setminus ( S_0 \cup S_2 )$ adjacent to both $S_2$ and $S_0$, and equip $(Q; S_0, S_1, S_2, S_3)$ with the structure of a topological quadrilateral, where the sides are indexed counterclockwise and $S_1, S_3$ lie on $\bdry \domainn$.
Let $q$ be the smallest number $> (1- \eps)$ such that $\confmapD_\n (q e^{i \theta}) \in S_0$ (such exists by the counter-assumption), and $p$ is the largest of number $<q$ such that $\confmapD_\n (p e^{i \theta}) \in S_2$.

\begin{figure}
\includegraphics[ width=0.34\textwidth]{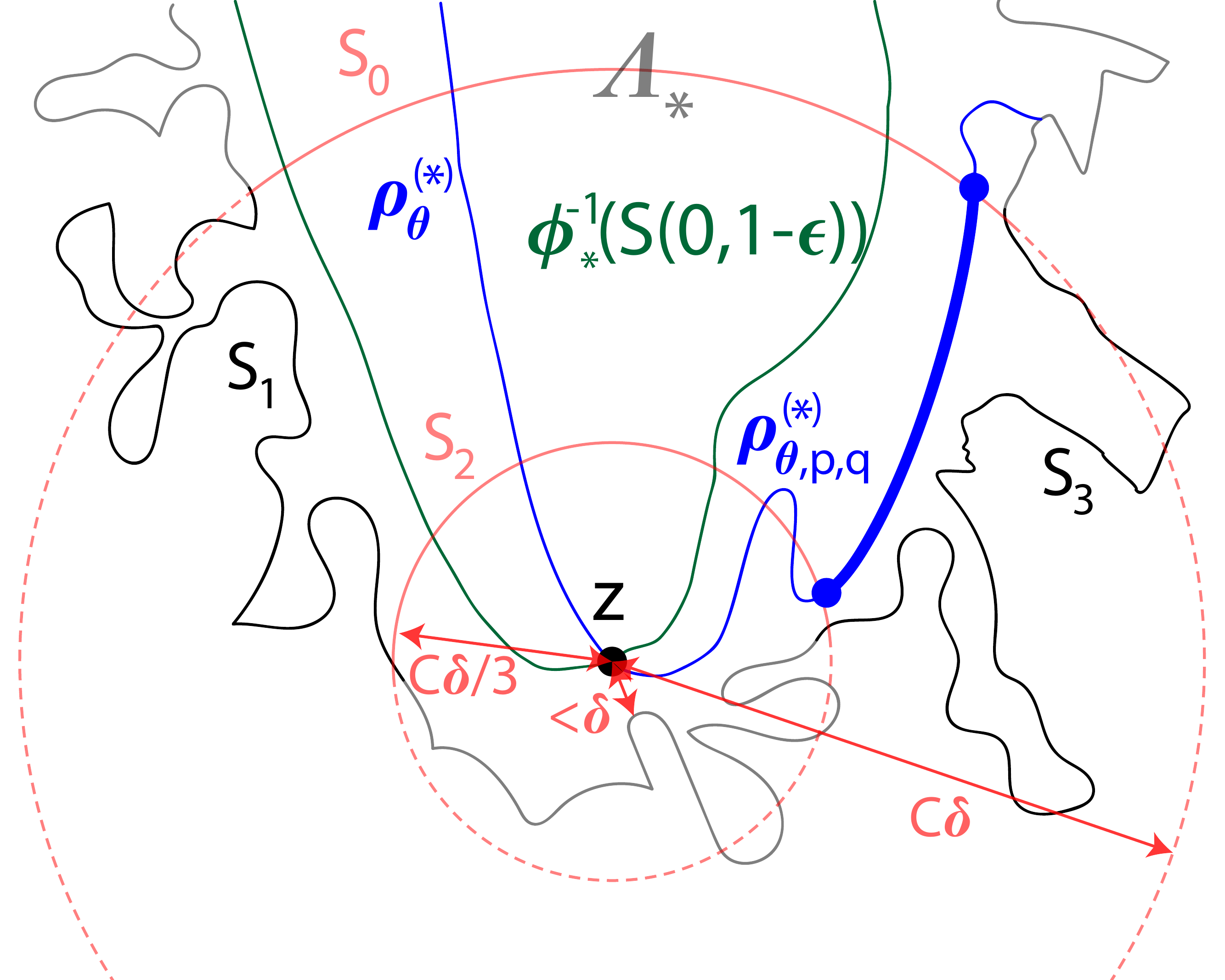}%
\caption{\label{fig: cross cuts and conf ray} The point $z$, with $d(z, \bdry \domainn) < \delta$, is the crossing point of the conformal image $\confmapD_\n (S(0, 1- \eps))$ of a circle and the conformal ray $\rho^{(\n)}_\theta$. The cross cuts $S_2$ and $S_0$ are arcs of the circles $S (z, C \delta/3)$ and $S (z, C \delta)$, respectively, and they determine a quadrilateral $Q \subset \domainn$  with sides $S_0, S_1, S_2$, and $ S_3$, crossed by 
$\rho^{(\n)}_{\theta, p, q}$.
}
\end{figure}

Denote $h^{(\n)} (z) = h (e^{-i\theta} \confmap_{\n} (z) )$ where $h$ is the harmonic function (measure)~\eqref{eq:def of hm}.
Using first the conformal invariance of harmonic measures, then a simple probabilistic Brownian motion argument, and finally the Beurling estimate and the assumption $d(z, \bdry \domainn ) < \delta$, we calculate
\begin{align}
\nonumber
h^{(\n)} (z) &= \HarmMeas (z, \domainn \setminus \rho^{(\n)}_{\theta, p, q}, \rho^{(\n)}_{\theta, p, q}) \\
\nonumber
& \le \HarmMeas (z, \domainn \cap B(z, C \delta / 3), \bdry B(z, C \delta / 3) )  \\
\label{eq: first harm meas estimate}
& \le A \left( \frac{d(z, \bdry \domainn)}{C \delta / 3} \right)^\alpha  \le A (3/C)^\alpha.
\end{align}
Note that if $C$ is large, this quantity is small, while on the other hand by Lemma~\ref{lem: maximization property of h_m} and conformal invariance, $z$ maximizes the harmonic measure $h^{(\n)}$ of $\rho^{(\n)}_{\theta, p, q}$ in $\confmap_{\n}^{-1} (\overline{B(0, 1-\eps)})$. The strategy of the proof is now, informally, to find $w \in \confmap_{\n}^{-1} (\overline{B(0, 1-\eps)})$ which is close to $\rho^{(\n)}_{\theta, p, q}$ and thus produce a contradiction.

\begin{figure}
\includegraphics[ width=0.34\textwidth]{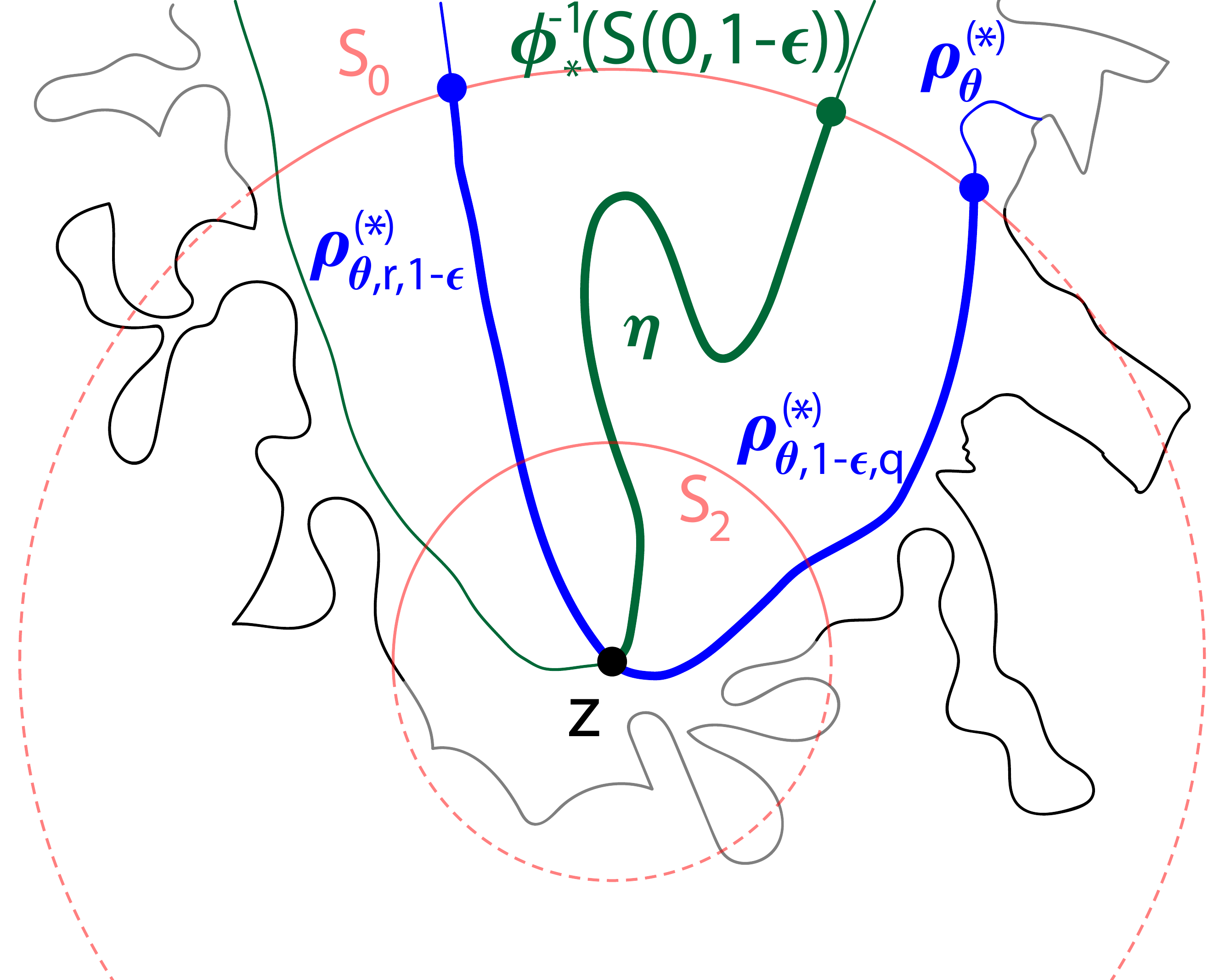} \qquad
\includegraphics[ width=0.34\textwidth]{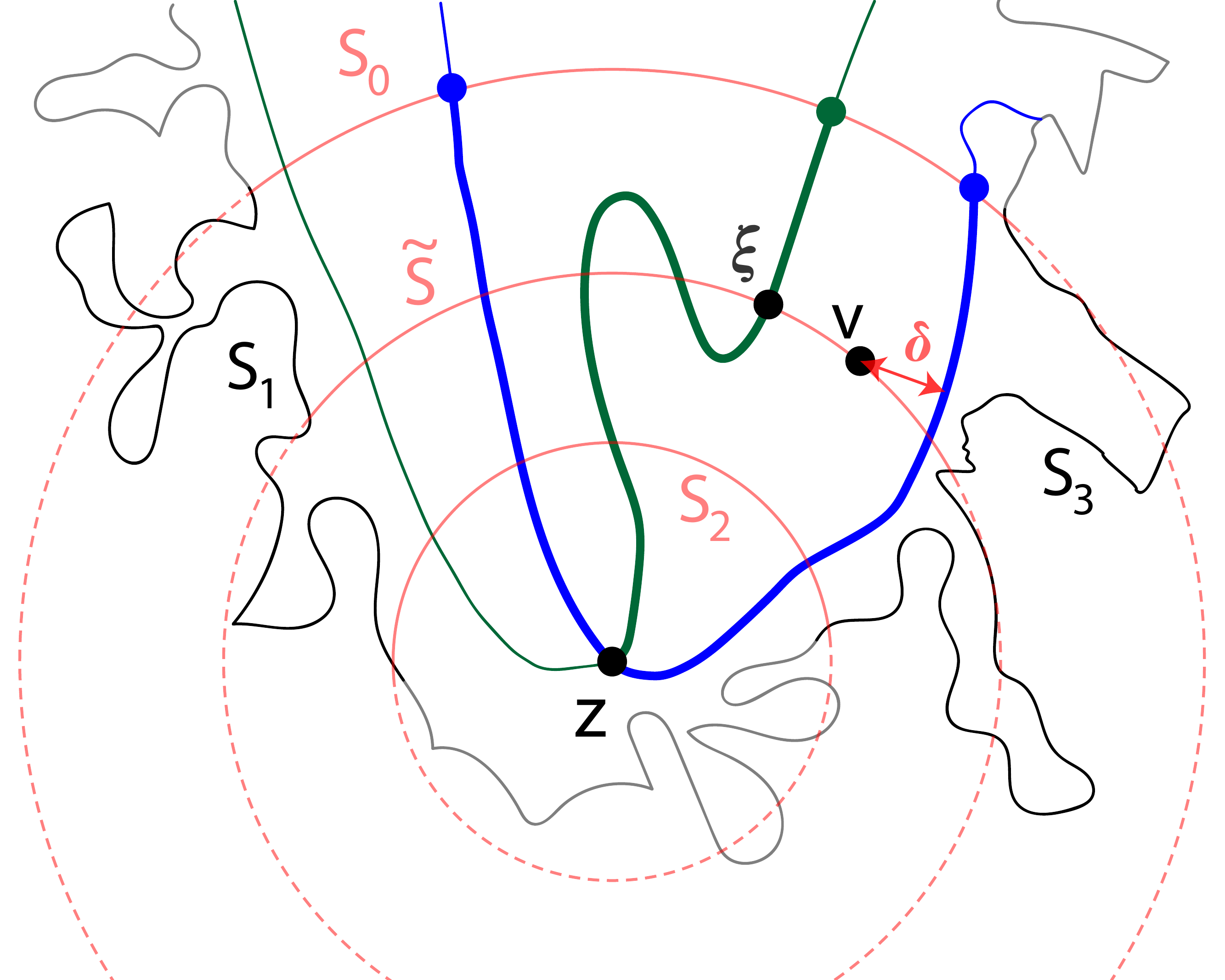}%
\caption{\label{fig: three curves} Left: the curves $\rho^{(\n)}_{\theta, r, 1-\eps }$, $\rho^{(\n)}_{\theta, 1-\eps, q}$, and $\curve$ from $z$ to $S_0$. Right: the cross cut $\tilde{S}$ and the points $\xi \in \curve \cap \tilde{S}$ and $v \in \tilde{S}$}
\end{figure}

To formalize this strategy, let $\tilde{S}$ be the innermost disconnecting arc of $S(z, 2 C \delta / 3)$, which by Lemma~\ref{lem: disconnecting arcs} traverses from side $S_1$ to $S_3$ inside the quadrilateral $Q$. We claim that there is a point $w \in \confmapD_\n ( S (0, 1 - \eps) )$ such that
\begin{align}
\label{eq: ppties of w}
d(w, \rho^{(\n)}_{\theta, p, q} ) \le 2 \delta \quad \text{and} \quad d(w, \tilde{S} ) \le \delta.
\end{align}
To prove this claim, we will study the following three curves from $z$ to $S_0$
(see Figure~\ref{fig: three curves}(Left)): (1) $\rho^{(\n)}_{\theta, 1-\eps , q}$, where $q$ is as previously; (2) $\rho^{(\n)}_{\theta, r, 1-\eps }$, where $r$ is the largest number $< 1 - \eps$ such that $\confmapD_\n (r e^{i \theta}) \in S_0$; and (3) the segment $\curve$ along the curve  $\confmapD_\n ( S (0 , 1- \eps ))$ from $z$ to first hitting $S_0$, to the direction chosen so that $\curve$ is disconnected from $\bdry \domainn$ by $S_0 \cup \rho^{(\n)}_{\theta, r , q} $. Note that $\eta$ has to cross $\tilde{S}$; if some point of $\eta \cap \tilde{S}$ is at a distance $\leq \delta$ from $\rho^{(\n)}_{\theta, p, q}$ we are done, so assume the contrary. Let then $\tilde{Q}$ be the connected component of $Q \setminus \rho^{(\n)}_{\theta, p, q}$ crossed by ${\curve}$. Note that either $S_1$ or $S_3$ is disconnected from $\tilde{Q}$ by $ \rho^{(\n)}_{\theta, p, q}$; for the rest of the proof we will assume for definiteness that it is $S_3$ as in Figure~\ref{fig: three curves} (if not, just reverse the roles of $S_1$ and $S_3$). Hence, $\tilde{Q}$ is a topological quadrilateral with $S_1$ and $ \rho^{(\n)}_{\theta, p, q}$ being two opposite sides and the two other ones contained in $S_0$ and $S_2$, respectively. Let $\sigma$ be a connected component of $\tilde{S}$ in $\tilde{Q}$ that disconnects $S_0$ from $S_2$ in $\tilde{Q}$, so ${\curve}$ crosses $\sigma$. Let $\xi \in \curve \cap \sigma$ be the last such intersection point along $\sigma$, when $\sigma$ is directed from $S_1$ to $\rho^{(\n)}_{\theta, p, q}$; by assumption $d(\xi, \rho^{(\n)}_{\theta, p, q}) > \delta$ (see Figure~\ref{fig: three curves}(Right)). Let finally $v$ be the first point, proceeding from $\xi$ towards $\rho^{(\n)}_{\theta, p, q}$ along the curve $\sigma$,  at which $d(v, \rho^{(\n)}_{\theta, p, q}) = \delta$. The proof of the claim will now be finished if we show that $\overline{B(v, \delta)}$ disconnects $S_0$ from $S_2$ in $\tilde{Q}$; indeed, a point $w \in {\curve } \cap \overline{B(v, \delta)}$ then exists and satisfies the desired properties. 

We thus prove this disconnection property of $\overline{B(v, \delta)}$. Let  $r'$ denote the smallest number $>r$ so that $\confmapD_\n (r' e^{i \theta}) \in S_2$, so $ \rho^{(\n)}_{\theta, r, r'}$ is a crossing of $\tilde{Q}$ in $\tilde{Q}$. Now, on the segment of the reversed curve $\sigma$ from $v$ up to first hitting $\confmapD_\n ( \overline{B (0, 1 - \eps) } ) $ (which occurs at latest in the point $\xi$), we have $ d (\cdot, \rho^{(\n)}_{\theta, p, q}) \ge \delta$ by definition, and $d (\cdot, S_0) ,  d (\cdot, S_2) \ge  C \delta / 3 \ge \delta$. But this curve lies in a quad $Q'$ restricted by $ \rho^{(\n)}_{\theta, p, q}$, $ \rho^{(\n)}_{\theta, r, r'}$, and two subcurves of $S_0$ and $S_2$, and the curve is disconnected in this quad from $ \rho^{(\n)}_{\theta, r, r'}$ by $\curve \subset   \confmapD_\n ( \overline{B (0, 1 - \eps) } )$. In particular, the curve hence connects $v$ to $  \confmapD_\n ( \overline{B (0, 1 - \eps) } ) $ $\delta$-inside $Q'$, and hence also $\delta$-inside $\domainn$. Since it was assumed that $B(0, 1-\eps)$ is a $\delta$-approximation of $\domainn$ under $\confmapD_{\n}$, it must thus hold that $d(v, \partial \domainn) < \delta$. But $d(v, \partial \domainn) \geq d(v, \partial \tilde{Q}) \geq \min \{ d(v, S_1), d(v, S_0), d(v, S_2), d(v, \rho^{(\n)}_{\theta, p, q})\}$, and $d(v, S_0), d(v, S_2), d(v, \rho^{(\n)}_{\theta, p, q}) \geq \delta$ as we saw above. Hence, we conclude that $ d(v, S_1) < \delta$ and that $\overline{B(v, \delta)}$ intersects the opposite sides $S_1$ and $\rho^{(\n)}_{\theta, p, q}$ of $\tilde{Q}$. This finishes the construction of the point $w$ satisfying~\eqref{eq: ppties of w}.

Define next $d^*$ as the following ``strong distance'' from $w$ to $S_1$:
$d^*$ is the supremum of $d \le ( C\delta /3 - \delta )$ such that the boundary of the connected component of $B(w, d)$ in $ \domainn \setminus \rho^{(\n)}_{\theta, p, q}$ that contains $w$, does not intersect $S_1$.
Hence, this connected component with $d=d^*$ does not intersect $S_0$, $S_1$, or $S_2$ and is separated from $S_3$ by $\rho^{(\n)}_{\theta, p, q}$. Simple harmonic measure arguments and the Beurling estimate thus yield
\begin{align}
\nonumber
h^{(\n)} (w) &= \HarmMeas (w, \domainn \setminus \rho^{(\n)}_{\theta, p, q}, \rho^{(\n)}_{\theta, p, q}) \\
\nonumber 
& \ge \HarmMeas (w,  B(w, d^* ) \setminus \rho^{(\n)}_{\theta, p, q} , \rho^{(\n)}_{\theta, p, q} )  \\
\nonumber 
& = 1-\HarmMeas (w,  B(w, d^* ) \setminus \rho^{(\n)}_{\theta, p, q} , S(w, d^* ))  \\
\label{eq: second harm meas estimate}
& \ge 1-A \big( d (w, \rho^{(\n)}_{\theta, p, q} ) / d^* \big)^{\alpha} \ge 1-A \left( 2 \delta/d^* \right)^{\alpha}.
\end{align}
Combining~\eqref{eq: first harm meas estimate},~\eqref{eq: second harm meas estimate}, and the maximization property $h^{ (\n)}(z) \ge h^{ (\n) }(w)$ of Lemma~\ref{lem: maximization property of h_m}, we obtain (assuming that $C$ is large enough so that $(3/C)^\alpha \leq 1/A$)
\begin{align}
\label{eq: d is bdd by an absolute cst}
d^* \le \left( 1/A - (3/C)^\alpha \right)^{- 1/ \alpha } 2 \delta,
\end{align}
so also the strong distance $d^*$ from $w$ to $S_1$ is also short; see Figure~\ref{fig: the point w}. In particular, denoting $d^{**} = \max\{ d^*, 2 \delta\}$, we observe that both $\rho^{(\n)}_{\theta, r , 1-\eps}$ and $\rho^{(\n)}_{\theta, p , q}$ must intersect $\overline{B(w, d^{**})}$.

\begin{figure}
\includegraphics[ width=0.34\textwidth]{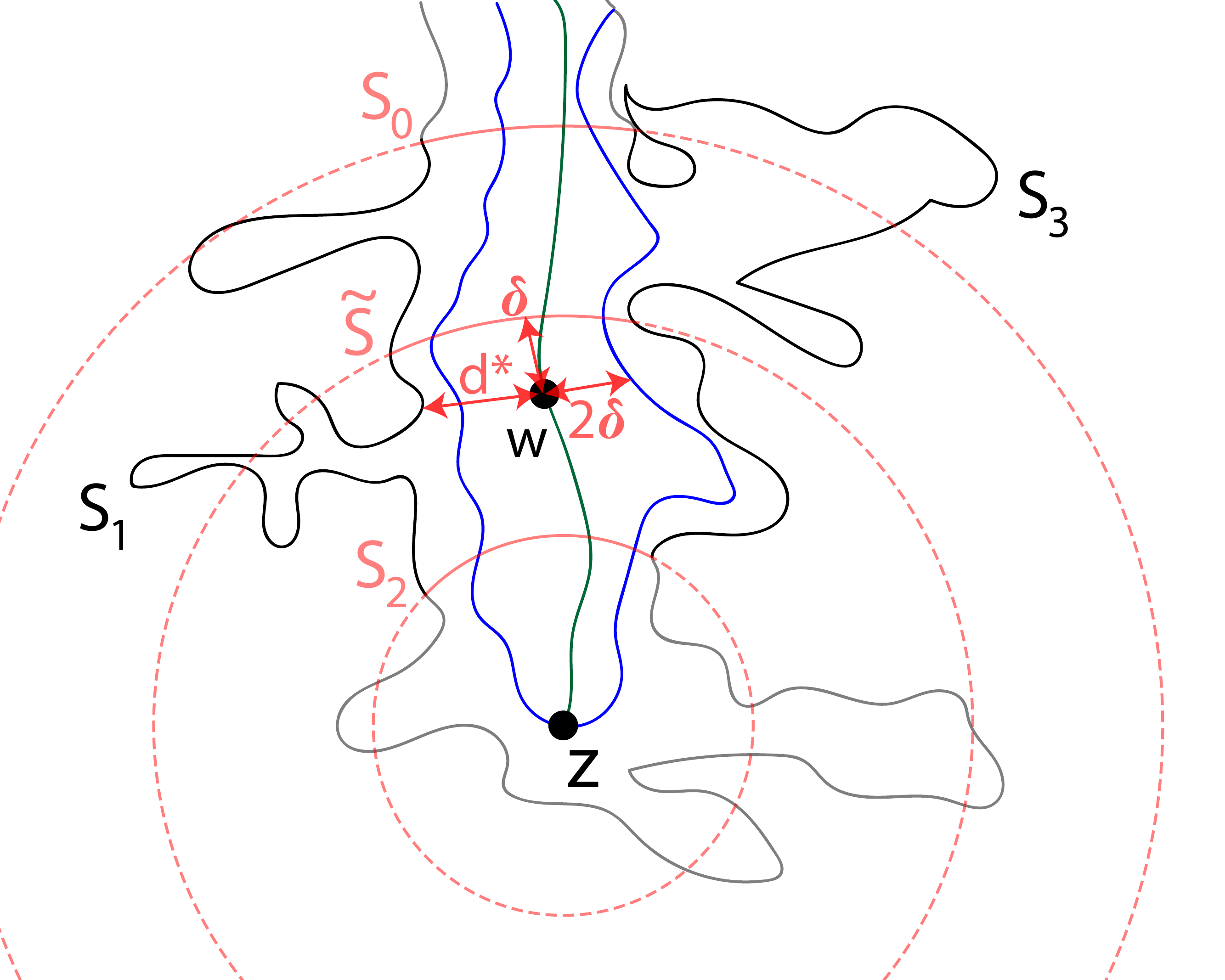}%
\caption{\label{fig: the point w} The point $w \in \curve$ satisfies $d(w, \rho^{(\n)}_{\theta, p, q} ) \le 2 \delta$, and $d(w, \tilde{S}) \le \delta$, and $d(w, S_1 ) \le d^* \le \left( 1/A - (3/C)^\alpha \right)^{- 1/ \alpha } 2 \delta$.}
\end{figure}

We can now conclude the contradiction. Namely, with the absolute constant $D$ small enough, a simple harmonic measure argument shows that the opposite conformal ray $\rho^{(\n)}_{\theta + \pi}$ cannot intersect the $2C \delta$-fjord with respect to $u$ defined by the cross cut $S_0$. In particular, it occurs that the entire boundary segment $S_1$ maps to the lower half-plane by $e^{-i\theta}\confmap_{\n}(\cdot)$, or that $S_3$ maps to the upper half-plane (or both). Let us assume for definiteness that the first case occurs; the second is treated analogously. The argument with Beurling estimate explained in Figure~\ref{fig:contradiction} and its caption now yields a contradiction. This concludes the proof. \hfill$\square$

\begin{figure}
\includegraphics[ width=0.22\textwidth]{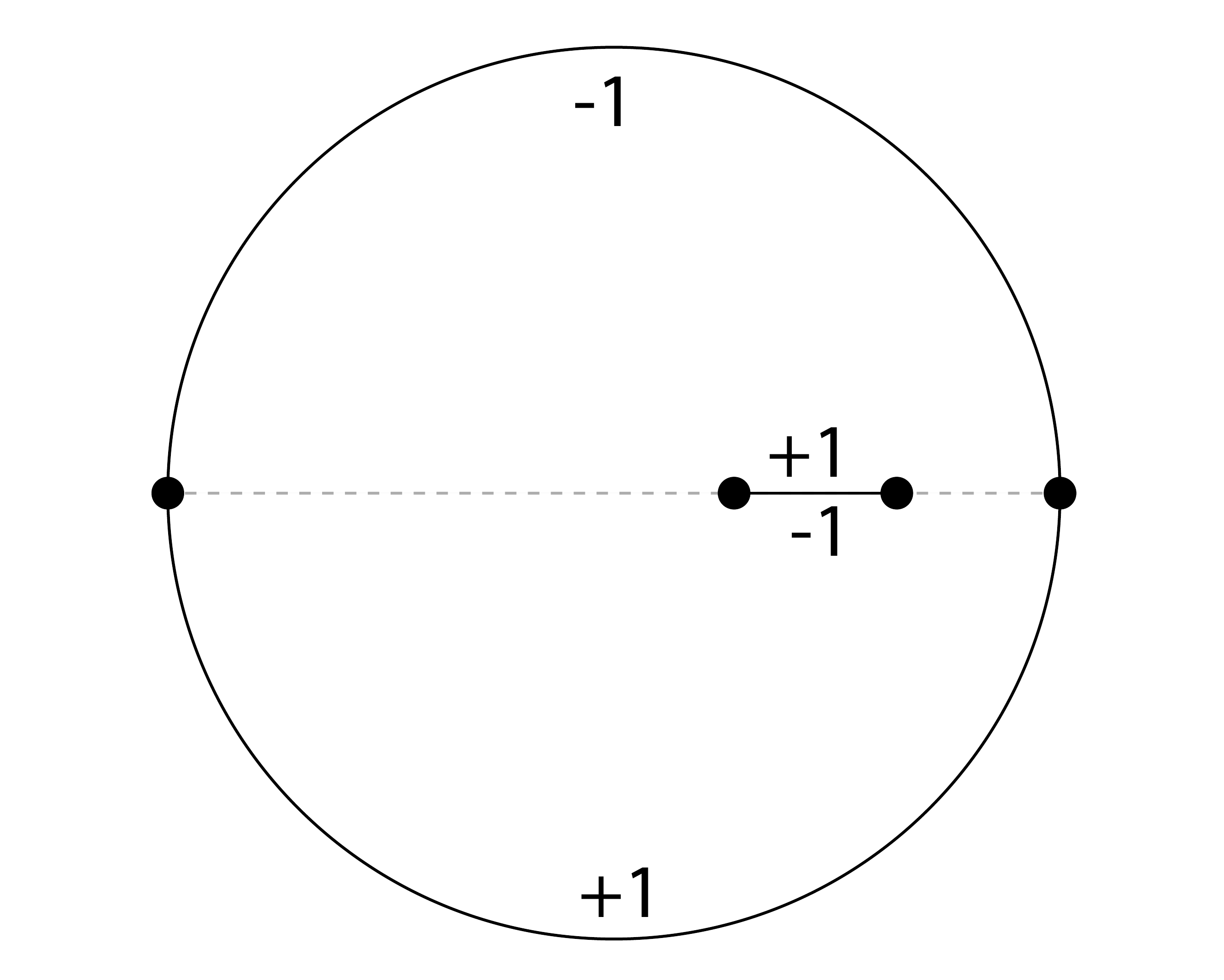} \quad
\includegraphics[ width=0.34\textwidth]{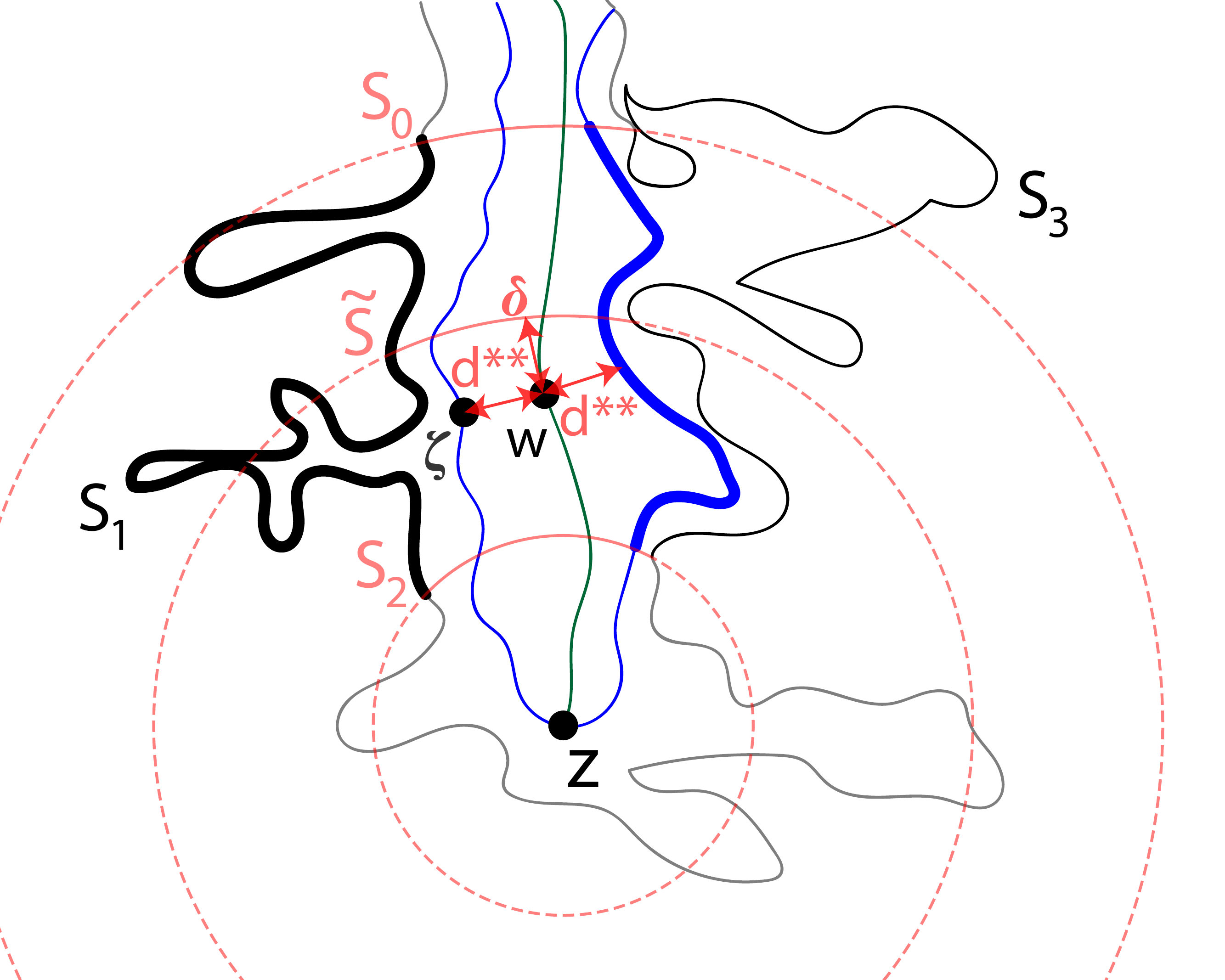}
\caption{
\label{fig:contradiction}
Left: The boundary values of a harmonic function $\tilde{h}$ (essentially a harmonic measure) on $\UnitD \setminus [p,q]$. By symmetry, this function vanishes on $(-1, 1) \setminus [p,q]$. Right: $\tilde{h}^{(\n)} (z) = \tilde{h} (e^{-i\theta}  \confmap_{\n} (z) )$ is a harmonic function on $\domainn \setminus \rho^{(\n)}_{\theta, p, q}$ that vanishes on $\rho^{(\n)}_{\theta,0, 1} \setminus \rho^{(\n)}_{\theta,p, q}$. On the other hand, there must exist a point $\zeta \in \rho^{(\n)}_{\theta,0, p}$ that is close to $w$, and $\tilde{h}^{(\n)}$ takes boundary values $1$ on $S_1$ and on the side of $\rho^{(\n)}_{\theta, p, q}$ facing towards $\zeta$ (drawn in thick lines). The Beurling estimate shows that, when $C$ is large enough, $\tilde{h}^{(\n)}(\zeta) >0$, a contradiction.
}
\end{figure}

\subsection{Concluding the proof of the Key Lemma}
\label{subsec:concluding key lem}

\subsubsection{\textbf{The collection of fjords}}
\label{subsubsec: defining the fjords}

Given $\delta > 0$ we define, for a domain $\domainn$,  
 a finite collection $(F^{(\n)}_i)_{i=1}^{m^{(\n)}}$ of interior-disjoint $(6C \delta)$-fjords with respect to $u$ in $\domainn$ as follows. Start with the plane and $\domainn$. Then, we draw, say in black, the square grid $C \delta \Z^2$, and the component of the $(C + 1)\delta$-interior of $\domainn$ containing $u$, denoted $G$, where $C$ is the absolute constant from Proposition~\ref{prop: conformal rays lie in fjords}.
  Then, we draw, say in red, the simple loop $\eta$ on the grid $C \delta \Z^2$ that stays inside $G$ and encloses $u$ and a maximal amount of squares. Hence, any square of $C \delta \Z^2$ which is not enclosed by $\eta$ but has at least one side on the loop $\eta$, must intersect $\C \setminus G$. For all such squares, we first draw in red the boundary of the square from the side(s) on the loop $\eta$ in both directions until it hits $\bdry G$. From those hitting points, we draw, still in red, a straight line segment to the closest point on $\bdry \domainn$, thus of length $(C + 1)\delta$. It is a simple exercise in plane geometry to show that two such line segments cannot intersect in $\domainn$. Now, the line segments drawn in red divide $\domainn $ into connected components, one of which is the interior of $\eta$ and the remaining ones are the desired fjords $(F^{(\n)}_i)_{i=1}^{m^{(\n)}}$, cut from $\domainn$ by cross cuts drawn in red, which we denote by $(S^{(\n)}_i)_{i=1}^{m^{(\n)}}$. By construction, the fjords $(F^{(\n)}_i)_{i=1}^{m^{(\n)}}$ are fjords with respect to $u$, and interior-disjoint and finitely many. The mouths of the fjords have a diameter at most $6C\delta$: indeed, each mouth consists of line segments from the square of $C \delta \Z^2$, and two other segments of length $(C+1) \delta$, and $\sqrt{2} C \delta + 2 (C+1) \delta < 6 C \delta$.

Note that we constructed above fjords with respect to $u$, while Proposition~\ref{lemma: Kemppainen's fjord lemma} addressed fjords with respect to $(a_\n, b_\n)$. This difference is handled as follows. Let $V_\n$ be a fixed neighbourhood of $a_\n$ in $\domainn$, such that it also is a fjord with respect to $u$ in $\domainn$, determined by a cross cut $S_\n$.
Let $\delta > 0$ be small enough so that for the $6C\delta$-thickening $H$ of $S_\n$, $a_\n \not \in H$ and furthermore the connected component of the $6C \delta$-interior of $\domainn$ containing $u$ intersects the component of $a_\n$ in $V_\n \setminus H$. 

\begin{lem}
\label{lem: fjords wrt different pts}
Given $V_\n$ and $\delta$ as above, if $x \in \domain_\n \setminus V_{\n }$ lies in one of the fjords $F^{(\n)}_i$, then that fjord is also a fjord with respect to $a_\n$.
\end{lem}

\begin{proof}
If $F^{(\n)}_i$ and $V_\n$ are disjoint, the claim is immediate. If $F^{(\n)}_i$ and $V_\n$ intersect, then $ F^{(\n)}_i \not \subset V_\n $ since by assumption $x \in F^{(\n)}_i \setminus V_\n$.
On the other hand, by construction, $F^{(\n)}_i$ does not intersect the $6 C \delta$-interior of $\domain_\n$, while $V_\n$ does by assumption, so $ V_\n \not \subset F^{(\n)}_i$. Thus, if the two fjords intersect, they must do it in such a manner that their mouths $S^{(\n)}_i$ and $S_\n$ cross.
Note that $S^{(\n)}_i$ has diameter $< 6C \delta$. Hence, $S^{(\n)}_i$ lies in the $6C \delta$-thickening of $S_\n$, i.e., $S^{(\n)}_i \subset H$. By the definition of $V_\n$, $a_\n$ is  connected in $\domainn \setminus H \subset\domainn \setminus S^{(\n)}_i$ to the component of the $6C \delta$-interior of $\domainn$ containing $u$ and, continuing within the latter, to $u$ in $\domainn \setminus S^{(\n)}_i$. Hence, $F^{(\n)}_i$ is a fjord with respect to $a_\n$.
\end{proof}
 
Define $P^{\domainn}_{\eps}: \domainn \to \domainn$ by $P^{\domainn}_{\eps} = \confmap_\n^{-1} \circ P_\eps \circ \confmap_\n$. The following (simple) lemma is the key in combining Propositions~\ref{lemma: Kemppainen's fjord lemma} and~\ref{prop: conformal rays lie in fjords}. 

\begin{lem}
\label{lem: projection is nice except in fjords}
Let $\domainn$ be a simply-connected planar domain and $\confmap_\n : \domainn \to \UnitD$ its Riemann map normalized at $u \in \domainn$, and $C, D > 0$ the absolute constants from Lemma~\ref{lem: conformal balls get close to boundary}.
Let $\delta \in (0, \tfrac{D}{C} d(u,\bdry \domainn ) )$ and suppose that $B(0, 1-\eps)$ is a $\delta$-approximation of $\domainn$ under $\confmapD_{\n}$.
Then, every $x \in \domainn$ with $P^{\domainn}_{\eps} (x)  \ne x$  lies in some of the fjords $F^{(\n)}_i$. If
$
d ( x, P^{\domain_\n}_{\eps} (x) ) \ge \ell,
$
then $x$ lies $( \ell - 7 C \delta )$-deep in that fjord.
\end{lem}

\begin{proof}
Suppose that $x \neq P^{\domainn}_{\eps} (x) $. 
By Proposition~\ref{prop: conformal rays lie in fjords}, $d(P^{\domainn}_{\eps} (x) , \bdry \domainn ) < \delta$ and the innermost disconnecting component $S$ of the circle arc $S(P^{\domainn}_{\eps} (x) , C \delta )$ also disconnects $x$ from $u$. Hence, neither $x$ nor $S$ can intersect the component of the $(C + 1)\delta$-interior of $\domainn$ that contains $u$, i.e., $G$, so by the definition of the fjords $F^{(\n)}_i$, $x$ lies in one of them, say $F^{(\n)}_j$, and $S$ in the union of them.

Suppose then that $d ( x, P^{\domainn}_{\eps} (x) ) \ge \ell$. If $S \subset F^{(\n)}_j$, we compute
\begin{align*}
d_{\domainn} (x, S^{(\n)}_j) & \ge d_{\domainn} (x, S)  \ge  d (x, S)  \ge d (x, P^{\domainn}_{\eps} (x) ) - C \delta \ge \ell - C \delta.
\end{align*}
If $S \not \subset F^{(\n)}_j$ note first that, by definition, all the fjords $F^{(\n)}_i$ intersect $G$ while $S$ does not. On the other hand, if $S \subset \domainn \setminus F^{(\n)}_j$ were to occur, since $x \in S \cap F^{(\n)}_j$, $S$ would have to disconnect the entire $F^{(\n)}_j$ from $u$ and in particular intersect $G$, a contradiction. Thus, $S \not \subset F^{(\n)}_j$ means that $S$ has to intersect $S^{(\n)}_j$, and in particular $d( P^{\domainn}_{\eps} (x) , S^{(\n)}_j) \le C \delta$. This gives
\begin{align*}
d_{\domainn} (x, S^{(\n)}_j) & \ge d (x, S^{(\n)}_j)  \ge  d (x, P^{\domain_\n}_{\eps} (x) ) - d(P^{\domain_\n}_{\eps} (x) , S^{(\n)}_j) - \diam(S^{(\n)}_j) \ge \ell - 7 C \delta.
\end{align*}
where the second step is the triangle inequality. The claim follows.
\end{proof}

\subsubsection{\textbf{Proof of Lemma~\ref{lem: key lemma}}}

Recall that the constants $\tilde{\eps}, \ell$ as well as the area bound $M$ were given in the setup of the lemma. Set the parameter $\chi$ through condition (G) so that, for any $n$, an unforced crossing of a fixed boundary annulus $A(z, \rho, \chi \rho)$ by the curve $\gamma^{(n)}$ occurs with probability $\le \tilde{\eps} / 3$ and then $r>0$ so that $(2\chi + 1) r < \ell$. The remaining parameters will be chosen so that first $\delta < \min \{ \delta^{(1)}_0(r), \delta^{(2)}_0 (M, \ell, \tilde{\eps}) \}$, then $\epsilon < \epsilon_0 (\delta)$, and finally $n> \max\{ n^{(1)}_0(r, \delta), n^{(2)}_0 (\epsilon, \delta) \}$, where the limiting parameter values are specified in the proof below.
 
 Let $S_r$ be the cross cut in the limit domain $\domain$ as in the definition of close approximations, equipped with an auxiliary reference point (see Section~\ref{subsubsec: close approximations}) and let $\delta^{(1)}_0(r)$ be small enough so that the assumption of (and right above) Lemma~\ref{lem: fjords wrt different pts} holds for the limiting domain $\domain$ and the fjord $V$ cut out by $S=S_r$ whenever $\delta < \delta^{(1)}_0$. By the Carath\'{e}odory convergence $\domain_n \Cara \domain$ and close approximation property, there exists $n^{(1)}_0 = n^{(1)}_0(r, \delta)$ so that for $n > n^{(1)}_0$ (i) $a_n \in B(a, r  ) $; (ii) there exists a corresponding cross cut $S^{(n)}$ in $\domain_n$; and (iii) the assumption of Lemma~\ref{lem: fjords wrt different pts} also holds for $\domain_n$ with the fjord $V^{(n)}$ cut out by $S^{(n)}$ and $\delta$.
 
By close approximation, the curves $\gamma^{(n)} $ are not forced to cross any component of $A (a, r , \chi  r ) $ separated from $u$ by $S^{(n)}$, and by the choice of $\chi$,
\begin{align*}
\PR^{(n)} [  x \in B(a, \chi r )  \text{ for all } x \in \gamma^{(n)} \cap V_n ] \ge 1- \tilde{\eps} / 3.
\end{align*}
Suppose now that the probable event above occurs and let $x \in \gamma^{(n)}  \cap V_n$. 
Let $\epsilon < \epsilon_0 (\delta)$ be small enough and $n > n^{(2)}_0 (\epsilon, \delta)$ large enough so that Lemma~\ref{lem: conformal balls get close to boundary} and thus also Proposition~\ref{prop: conformal rays lie in fjords} holds true. The innermost disconnecting arc of $S(P^{\domain_n}_\eps (x  ), C\delta)$ thus defines a fjord that intersects $V_n$ at $x$; this arc hence either lies inside the fjord $V_n$ or it crosses its mouth. In the latter case,  $d( P^{\domain_n}_\eps (x  ) , B(a, r) ) \leq d( P^{\domain_n}_\eps (x  ) , S^{(n)} ) \le C \delta$ and thus $d(a, P^{\domain_n}_\eps (x  ) )  \leq C \delta + r$ and finally $d(x, P^{\domain_n}_\eps (x  ) ) \leq d(x, a) + d(a, P^{\domain_n}_\eps (x  ) )  \le \chi r + r +  C \delta < \ell$ (where $ C \delta  \le r $ by the assumptions of Lemma~\ref{lem: fjords wrt different pts}). In the former case, note that the curve $\gamma^{(n)}$ (by Proposition~\ref{prop: conformal rays lie in fjords}) after the visiting $x$ has to cross the innermost disconnecting arc of $S(P^{\domain_n}_\eps (x  ), C\delta)$ to later reach $b_n$, and this crossing point lies in $V_n$, so by the assumed probable event, the crossing point also lies in $B(a, \chi r )$. Thus, $d( P^{\domain_n}_\eps (x  ) , B(a, \chi r ) ) \le C \delta$ and with the same computation as above $d(x, P^{\domain_n}_\eps (x  ) ) \le ( 2\chi + 1) r < \ell$.
In conclusion,
\begin{align}
\label{eq:bdry point}
\PR^{(n)} [  d(x,P^{\domain_n}_\eps (x  ) ) < \ell \text{ for all } x \in \gamma^{(n)} \cap V_n ] \ge 1- \tilde{\eps} / 3.
\end{align}

The boundary point $b_n$ is treated similarly; let $\tilde{V}_n$ in what follows denote the counterpart of $V_n$ for $b_n$. 

Suppose now that $x \in \Lambda_n \setminus (V_n \cup \tilde{V}_n)$, and $d(x, P^{\domain_n}_\eps (x  )) \geq \ell$. By Lemmas~\ref{lem: fjords wrt different pts}--\ref{lem: projection is nice except in fjords}, $x$ then lies $(\ell-7C \delta)$-deep in one of the fjords $F^{(n)}_i$, which also is a fjord with respect to $(a_n, b_n)$.
Hence,
\begin{align}
\nonumber
\PR^{(n)} & [ d(x,P^{\domain_n}_\eps (x  ) ) \geq \ell \text{ for some } x \in \gamma^{(n)} \cap (\domain_n \setminus (V_n \cup \tilde{V}_n)) ]
\\ 
\label{eq:bulk point}
& \leq 
\PR^{(n)} [ \gamma^{(n)} \text{ visits $\ell -  7 C \delta$-deep in some of the fjords $(F^{(n)}_i)_{i=1}^{m^{(n)}}$}] 
 \leq \tilde{\epsilon}/3,
\end{align}
where on the last line, we chose by Proposition~\ref{lemma: Kemppainen's fjord lemma} a $\delta^{(2)}_0 = \delta^{(2)}_0 (M, \ell, \tilde{\eps})$ so that for all $\delta < \delta^{(2)}_0$, the event on the left has probability $\le \tilde{\eps}/3$. 

By combining~\eqref{eq:bdry point}, its analogue for $b_n$, and~\eqref{eq:bulk point}, we now obtain
\begin{align*}
\PR^{(n)} & [ d(x,P^{\domain_n}_\eps (x  ) ) \geq \ell \text{ for some } x \in \gamma^{(n)} \cap \domain_n ] \leq \tilde{\eps}.
\end{align*}
So far, this only concerns \textit{interior points} $x \in \gamma^{(n)} \cap \domain_n$. To handle boundary points, note that by the assumed setup, $\gamma^{(n)}$ never travels a positive distance along $\partial \domain$. The statement of the Key lemma then follows by the continuity of the curves $\gamma^{(n)}$ and  $\confmap_n^{-1} \circ P_{\eps} ( \tilde{\gamma}_{\UnitD}^{(n)}  )$.
 \hfill $\square$

\appendix

\section{Some postponed (easy) proofs}
\label{app:roskis}

For the sake of completeness, we give here the proofs of some elementary results, related to conformal and harmonic maps, that were used in the bulk of the note.

\begin{lem}
\label{lem: disconnecting arcs}
Let $\domainn$ be a simply-connected domain and $z, u \in \domainn$ with $d(z, \bdry \domainn) <  d(z, u)$ and $d \in (d(z, \bdry \domainn) , d(z, u))$. Then, there are finitely many connected components of the circle arc $S(z, d)$ in $\domainn$ that disconnect $z$ from $u$ in $\domainn$. A unique one of these disconnecting components is \emph{innermost}, in the sense that it disconnects all the others from $z$ in $\domainn$. If $d_1 < d_2$ and $S_1$, $S_2$ are the innermost disconnecting arcs of $S(z, d_1)$ and $S(z, d_2)$, respectively, then $S_1$ separates $S_2$ from $z$ in $\domainn$.
\end{lem}

\begin{proof}
There exists a broken line of finitely many line segments from $u$ to $z$ in $\domainn$. Such a broken line intersects $S(z, d)$  finitely many times, proving the finite number of separating components. The uniqueness of the innermost component  follows since the separating components $S^{(j)}$ of $S(z, d)$ are disjoint, and thus each $S^{(j)}$ falls into one connected component of $\domainn \setminus S^{(i)}$, where $i \ne j$. A similar disjointness argument shows the ordering of $S_1$ and $S_2$.
\end{proof}

\begin{proof}[Proof of Lemma~\ref{lem: conformal balls get close to boundary}]
Denote $K_\eps := \confmapD ( \overline{B (0 , 1 - \eps ) } )$, $K^{(n)}_\epsilon := \confmapD_{n} ( \overline{B (0 , 1 - \eps ) } )$ and $
\domain_{r} := \{  z \in \domain: \; d(z, \bdry \domain) \ge r  \}
$. %; note that these domains are compact.
 Fix $\epsilon_0 (\delta)$ so that $\domain_{\delta / 2} \subset K_\epsilon$ if $\epsilon < \epsilon_0$. Given $\epsilon < \epsilon_0$, denote $d :=d(K_\epsilon, \bdry \domain ) \leq \delta / 2$ and fix $n_0 (\epsilon, \delta)$ so that for all $n > n_0$: (a) $|\confmap_n^{-1} (\cdot)- \confmap^{-1} (\cdot)| <  d/2 \leq \delta /4$ on $B(0, 1-\epsilon)$ and (b) for all $w \in \bdry \domain$ there exists $w_n \in \bdry \domain_n$ with $|w-w_n|< d/4 \leq \delta /8$.
Now, for claim~(i), (a) above implies $\domain_{3 \delta /4} \subset K^{(n)}_\epsilon \subset \domain_{d/2}$ so by (b) $d(\zeta, \bdry \domain_n) <7\delta/8$. For claim~(ii), suppose first that $z \in \domain $. By assumption, then $z \in \domain \setminus K^{(n)}_\epsilon$ and from the above $K^{(n)}_\epsilon \supset \domain_{3 \delta /4}$ so $d(z, \partial \domain) < 3 \delta /4$ and by (b) $d(z, \bdry \domain_n) <7\delta/8$, a contradiction. Thus, we must have $ z \not \in \domain $, and any path $\eta$ from $z$ to $K^{(n)}_\epsilon$ thus has to intersect $\bdry \domain$. Denote $m=\min \{ d, d(z, \bdry \domain) \} / 8 > 0$ and let $\xi$ be the first point on the path $\eta$ (directed from $z$ to $K^{(n)}_\epsilon$) with $d(\xi, \bdry \domain) = m$; hence by property (b), $d(\xi, \bdry \domain_n) \leq 3d/8 $
and on the other hand
$d(\xi, K^{(n)}_\epsilon) \geq d(\bdry \domain, K^{(n)}_\epsilon) \geq d/2$ since $\bdry \domain$ separates $\xi$ from $K^{(n)}_\epsilon$. A point $u \in \bdry \domain_n$ with $d(\xi, u) \leq 3d/8 $ thus exists and cannot be separated from $\xi$ in $G_\delta$ by $K^{(n)}_\epsilon$. This concludes the proof.
\end{proof}

\begin{proof}[Proof of Lemma~\ref{lem: maximization property of h_m}]
Denote $\UnitD \setminus [p, q] = U$ for short. Green's third identity states that
\begin{align}
\label{eq: Green's third}
h (z) = \int_{w \in \bdry  U}  ( h (w) \nabla_w G(z, w) - G(z, w) \nabla_w  h (w)  ) \cdot \nu (w) \vert dw \vert \quad \text{for any } z \in U,
\end{align}
where $\nu (w)$ is the outward normal unit vector of $U$ at $w$, $\vert dw \vert$ denotes the length element along the boundary $\bdry  U$, and $G(z, w)$ is the Green's function of the Laplacian in any domain containing $U$. We choose the Green's function in $\UnitD$,
\begin{align*}
G(z, w) = \tfrac{1}{2 \pi} \log \vert \tfrac{z - w}{1 - z w^*} \vert.
\end{align*}

Now, the first term $h(w) \nabla_w G(z, w) \cdot \nu (w) $ on the right-hand side of \eqref{eq: Green's third} cancels out: $h (w) = 0$ on $\bdry \UnitD$, while $[p, q]$ is integrated in two directions with opposite normals $\nu (w)$. The second term disappears on $\bdry \UnitD$, leaving
\begin{align}
\label{eq: Green's third 2}
h (z) = - 2 \int_{w = p}^q  G(z, w) \vert \nabla_w  h (w)  \vert dw,
\end{align}
where we combined the integrations in two directions on the line segment $[p,q]$, using the facts that $\nabla_w  h (w)$, on either side of the segment, points outward of the domain $U$ and that its modulus  $\vert \nabla_w  h (w)  \vert$ is well defined even if the direction is not. Finally, examine the M\"{o}bius map $\UnitD \to \UnitD$ given by $z \mapsto \frac{z - w}{1 - z w^*} $, where $w \in [p, q]$: it preserves $\R \cap \UnitD$ (with orientation), mapping $z=w$ to the origin and the sphere $S(0, 1- \eps)$ to another $\R$-centered sphere that lies entirely left of the origin. Hence, for any $w \in [p, q]$, the function $G(z, w) = \frac{1}{2 \pi} \log \vert \frac{z - w}{1 - z w^*} \vert$ attains its unique minimum over $z \in \overline{B(0,1-\eps)}$ at $z = 1 - \eps$. The claim now follows from~\eqref{eq: Green's third 2}.
\end{proof}

\end{document}